\documentclass{article}

\usepackage[english]{babel}
\usepackage{fullpage}
\usepackage{authblk}
\usepackage{hyperref}
\usepackage{booktabs,tabularx}

\usepackage{lineno}

\usepackage{amsmath}
\usepackage{amsthm}
\usepackage{amssymb}
\usepackage{graphicx}
\usepackage{mathtools}

\usepackage{tikz}
\usetikzlibrary{calc}
\usetikzlibrary{hobby}

\newtheorem{lemma}{Lemma}
\newtheorem{theorem}{Theorem}

\newtheorem{definition}{Definition}
\newtheorem{claim}{Claim}

\usepackage{thmtools,thm-restate}

\newcommand{\ghom}[1]{\textsc{Hom}(#1)}
\newcommand{\plainghom}{\textsc{Graph Homomorphism}}
\newcommand{\subiso}[1]{\textsc{Subgraph Isomorphism}(#1)}

\newcommand{\VCSP}[1]{\textsc{VCSP}(#1)}
\newcommand{\plainVCSP}{\textsc{VCSP}}
\newcommand{\plainCSP}{\textsc{CSP}}
\newcommand{\valhom}[1]{\textsc{ValHom}(#1)}

\newcommand{\majcover}[1]{\text{pmn}(#1)}

\newcommand{\struct}[1]{\mathbf{#1}}

\newcommand{\feas}[1]{\text{Feas}(#1)}
\newcommand{\supp}[1]{\text{supp}(#1)}
\newcommand{\opt}[1]{\text{opt}(#1)}
\newcommand{\cost}[1]{\text{cost}(#1)}

\newcommand{\sa}[3]{\textbf{SA}_{#1,#2}(#3)}
\newcommand{\comp}[1]{\text{Comp}(#1)}

\newcommand{\problemStatement}[3]{%
  \begin{center}
  \begin{tabularx}{\columnwidth}{@{}lX@{}}
  \toprule
  \multicolumn{2}{@{}c@{}}{\textsc{#1}}\tabularnewline
  \midrule
  \bfseries Input:    & #2 \\
  \bfseries Goal: & #3 \\
  \bottomrule
  \end{tabularx}
  \end{center}
}

\title{Faster algorithms for graph homomorphism via tractable constraint satisfaction}
\author{Clément Carbonnel}
\affil{CNRS, LIRMM, University of Montpellier}
\date{}

\begin{document}
\maketitle

\begin{abstract}
    We show that the existence of a homomorphism from an $n$-vertex graph $G$ to an $h$-vertex graph $H$ can be decided in time $2^{O(n)}h^{O(1)}$ and polynomial space if $H$ comes from a family of graphs that excludes a topological minor. The algorithm is based on a reduction to a single-exponential number of constraint satisfaction problems over tractable languages and can handle cost minimization. We also present an improved randomized algorithm for the special case where the graph $H$ is an odd cycle.
\end{abstract}

\section{Introduction}

A homomorphism from a graph $G$ to a graph $H$ is a mapping $g : V(G) \to V(H)$ such that $\{g(u),g(v)\}$ is an edge of $H$ whenever $\{u,v\}$ is an edge of $G$. The $\plainghom$ problem takes a pair of graphs $(G,H)$ as input and asks whether there exists a homomorphism from $G$ to $H$.  
This problem is NP-complete and can be solved in time $2^{O(n \log h)}$ by exhaustive enumeration of all possible mappings from $V(G)$ to $V(H)$, where $n = |V(G)|$ and $h = |V(H)|$. Cygan et al.~\cite{DBLP:journals/jacm/CyganFGKMPS17} proved that no algorithm can solve $\plainghom$ in time $2^{o(n \log h)}$ unless the ETH fails, providing strong evidence that no significant improvement is possible over the brute-force algorithm.

Given a family of graphs ${\mathcal H}$, the problem $\ghom{\mathcal H}$ is the restriction of $\plainghom$ to inputs $(G,H)$ where $H \in \mathcal H$. If ${\mathcal H}$ is composed of a single graph $H$, a theorem of Hell and Ne\v{s}et\v{r}il~\cite{DBLP:journals/jct/HellN90} states that either $H$ is bipartite and $\ghom{\mathcal H}$ can be solved in polynomial time, or $H$ contains an odd cycle and $\ghom{\mathcal H}$ is NP-complete. Finding a homomorphism to a bipartite graph is a trivial problem (equivalent to finding a vertex coloring of $G$ with two colors), so this theorem implies that $\ghom{\mathcal H}$ does not have interesting polynomial-time cases. On the other hand, there is some variability when $\ghom{\mathcal H}$ is NP-complete: certain of these problems can be solved in time $2^{O(n)} h^{O(1)}$ (which is the best runtime one can hope for under the ETH, based on a result of Jonsson, Lagerkvist, and Roy~\cite{DBLP:journals/toct/JonssonLR21}), while for many others no algorithm with single-exponential dependency on $n$ and $h$ is known. This motivates a more finely grained complexity analysis, in which the goal is to characterize the families of graphs $\mathcal H$ for which $\ghom{\mathcal H}$ admits a single-exponential algorithm. The possibility of such a characterization was raised as an open problem by Cygan et al.~\cite{DBLP:journals/jacm/CyganFGKMPS17}.

The most prominent results in this line of research concern $\mathcal K$, the family of all complete graphs. $\ghom{\mathcal K}$ corresponds to the graph coloring problem where the number of colors is part of the input. Lawler~\cite{DBLP:journals/ipl/Lawler76} observed that $\ghom{\mathcal K}$ can be solved in time $O(2.443^n)$ via dynamic programming. After a series of improvements, an algorithm with runtime $2^nn^{O(1)}$ was presented by Björklund, Husfeldt, and Koivisto~\cite{DBLP:journals/siamcomp/BjorklundHK09} based on the inclusion-exclusion principle. This algorithm is the fastest to date, although Björklund et al.~\cite{DBLP:conf/soda/BjorklundCHKP25} recently showed that $\ghom{\mathcal K}$ could be solved in time $O(1.99982^n)$ conditionally on a conjecture of Strassen~\cite{Strassen1994}.

For less specific families of graphs, Fomin, Heggernes, and Kratsch~\cite{DBLP:journals/mst/FominHK07} proved that if $\mathcal H$ has bounded treewidth then $\ghom{\mathcal H}$ can be solved in time $2^{O(n)}h^{O(1)}$. This result was generalized to all families of graphs of bounded cliquewidth by Wahlström~\cite{DBLP:journals/mst/Wahlstrom11} at the cost of an exponential dependency on $h$, and then to all families of graphs of bounded extended cliquewidth by Bulatov and Dadsetan~\cite{DBLP:conf/icalp/BulatovD20}. (Bounded extended cliquewidth is a more general property than bounded cliquewidth that applies, for example, to grid and hypercube graphs.) These three algorithms are based on a form of dynamic programming; as such, they require exponential space but can also solve various counting and optimization problems related to $\ghom{\mathcal H}$. Orthogonally, it is a folklore result that $\ghom{\mathcal H}$ can be solved in time $2^{O(n)}h^{O(1)}$ by a standard branching algorithm if $\mathcal H$ has bounded maximum degree.

\paragraph{Our contributions.} Our main contribution is an algorithm that solves $\ghom{\mathcal H}$ in single-exponential time and polynomial space if $\mathcal H$ excludes at least one topological minor. Excluding a topological minor is a hereditary property that generalizes both bounded treewidth and bounded maximum degree, while being incomparable in general with bounded (extended) cliquewidth.

The central component of our algorithm is a polynomial-time procedure that colors the vertex set of $H$ with a constant number $k$ of colors such that, for any vertex coloring of $G$, the problem of finding a color-preserving homomorphism from $G$ to $H$ can be solved in polynomial time. Since any homomorphism from $G$ to $H$ is color-preserving for at least one coloring of $G$, we can then solve the general homomorphism problem in single-exponential time and polynomial space by exhaustive enumeration of the $2^{O(n)}$ possible colorings of $G$ with at most $k$ colors.

For computing an appropriate coloring of $H$, our starting point is a simple observation that connects track layouts to tractable constraint satisfaction problems. A track layout of a graph is a partition of its vertex set into linearly ordered independent sets (tracks), with the constraint that edges with endpoints in the same tracks do not cross. Considering the vertex coloring of $H$ induced by a track layout, this topological constraint ensures that the color-preserving homomorphism subproblems can be formulated as constraint satisfaction problems of bounded relational width. This property causes a standard, polynomial-sized linear relaxation of the subproblems to have integrality gap $1$.

An influential theorem of Dujmovi\'c et al.~\cite{DBLP:journals/jacm/DujmovicJMMUW20} implies that all graphs that exclude a fixed minor admit a track layout with a constant number of tracks. Their proof is constructive and this layout can be computed in polynomial time. Combined with our observation, this theorem provides a single-exponential algorithm for $\ghom{\mathcal H}$ when $\mathcal H$ excludes a minor.  However, there exist families of graphs that exclude a topological minor but do not admit a track layout with a constant number of tracks. A canonical example (due to Wood~\cite{DBLP:journals/dmtcs/Wood08}) is the family of all graphs of maximum degree $3$. We address this issue by introducing a universal-algebraic generalization of track layouts, \emph{persistent majority colorings}, that is also bounded on graphs of bounded maximum degree and always induces color-preserving homomorphism subproblems of bounded relational width. We then prove that all families of graphs that exclude a topological minor admit polynomially computable persistent majority colorings with a bounded number of colors, using a structure theorem of Grohe and Marx~\cite{DBLP:journals/siamcomp/GroheM15} and adapting an argument of Dujmovi\'c, Morin and Wood~\cite{DBLP:journals/jct/DujmovicMW17}.

We state our result on a more general homomorphism problem that we call $\valhom{\mathcal H}$. In this variant, $G$ and $H$ are directed graphs and the underlying graph of $H$ belongs to $\mathcal H$. Each pair of arcs $(a_1,a_2) \in A(G) \times A(H)$ is associated with a (possibly infinite) cost $\eta(a_1,a_2)$. The goal is to find a homomorphism $g$ from $G$ to $H$ that minimizes the sum of the costs of mapped arcs, $\sum_{(u,v) \in A(G)} \eta((u,v),(g(u),g(v)))$. In the theorem below, $|\eta|$ denotes the encoding size of the cost function $\eta$.

 \begin{restatable}{theorem}{thmtopo}
    \label{thm:valhom}
    If $\mathcal H$ is a family of graphs that excludes at least one topological minor, then $\valhom{\mathcal H}$ can be solved in time $2^{O(n)} (h+|\eta|)^{O(1)}$ and polynomial space.
\end{restatable}

Various homomorphism problems on $\mathcal H$ can be reduced to $\valhom{\mathcal H}$ in polynomial time, with no increases to either $n$ or $h$. These include for example the problem of determining whether there exists a list homomorphism from a graph to another, or the vertex-based minimum-cost homomorphism problem on directed graphs introduced by Gutin et al.~\cite{DBLP:journals/dam/GutinRYT06}. Finally, the machinery behind Theorem~\ref{thm:valhom} is not limited to excluded topological minors: it provides single-exponential algorithms for all families of graphs with bounded persistent majority number, and some of those families do not exclude a topological minor. Further discussion on this topic, as well as formal definitions, can be found in Section~\ref{sec:dmn}.

Our second contribution is a randomized algorithm for $\ghom{C_{2k+1}}$, where $C_{2k+1}$ is the cycle on $2k+1$ vertices. The fine-grained complexity of this problem is of particular interest due to the fundamental role played by odd cycles in the Hell-Ne\v{s}et\v{r}il theorem. A deterministic algorithm of Fomin, Heggernes, and Kratsch~\cite{DBLP:journals/mst/FominHK07} solves $\ghom{C_{2k+1}}$ in time $(c_k)^n n^{O(1)}$, where $c_k = (2k+1)^{\frac{1}{2k+1}}(\frac{2k+1}{2k})^{\frac{2k}{2k+1}}$. For small values of $k$, specialized algorithms from Beigel and Eppstein~\cite{DBLP:journals/jal/BeigelE05} and Fomin, Heggernes and Kratsch~\cite{DBLP:journals/mst/FominHK07} yield the improved complexity bounds $O(1.329^n)$ for $k=1$ and $(\sqrt{2})^n n^{O(1)}$ for $2 \leq k \leq 4$. 

The $(c_k)^n n^{O(1)}$ algorithm of Fomin, Heggernes, and Kratsch boils down to the following observation: for any vertex $v \in V(C_{2k+1})$, the restricted list homomorphism problem in which all lists are either $\{v\}$ or $V(C_{2k+1}) \setminus \{v\}$ can be solved in polynomial time. The algorithm then solves $\ghom{C_{2k+1}}$ by enumerating all possible assignments of lists to vertices and solving the corresponding list homomorphism subproblems. The automorphisms of $C_{2k+1}$ ensure that only assignments that map at most $n/(2k+1)$ vertices to list $\{v\}$ need to be enumerated, which yields the claimed complexity bound.

We show that this idea can be further refined. Concretely, we construct a simple family $\mathcal S$ of three subsets of $V(C_{2k+1})$ that induces a tractable list homomorphism problem. These subsets are large and have significant pairwise overlaps, so assigning a list from $\mathcal S$ to a vertex of $G$ at random has a fairly high chance to be consistent with a hypothetical homomorphism. With an appropriate choice of probability distribution over $\mathcal S$, sampling approximately $(2^{-1/(2k+1)} \cdot c_k)^n$ tractable subproblems is sufficient to determine with high certainty whether there exists a homomorphism from $G$ to $C_{2k+1}$.

 \begin{restatable}{theorem}{thmodd}
    \label{thm:randodd}
    There exists a randomized algorithm that solves $\ghom{C_{2k+1}}$ in time $(\alpha_k)^n n^{O(1)}$ and polynomial space, where $\alpha_k = \left(\frac{2k+1}{2} \right)^{\frac{1}{2k+1}} \left( \frac{2k+1}{2k} \right)^{\frac{2k}{2k+1}}$.
\end{restatable}

These bounds improve upon previous results for all $k \geq 3$ at the cost of randomization. Numerically, we have $\alpha_3 < 1.365$, $\alpha_4 < 1.313$, $\alpha_5 < 1.274$ and $\alpha_6 < 1.244$. The previously best exponent bases for the corresponding values of $k$ were $\sqrt{2} > 1.414$ for $k\in\{3,4\}$, $c_5 > 1.356$ for $k = 5$, and $c_6 > 1.311$ for $k = 6$. 

\section{Preliminaries}

If $k$ is a positive integer and $M$ is a multiset, then $k \times M$ denotes the multiset obtained from $M$ by increasing the multiplicity of each element by a factor of $k$. We use the double braces notation $\{\{ \cdot \}\}$ for multisets and $[k]$ as shorthand for the set $\{1,\ldots,k\}$. The symbol $\overline{\mathbb{Q}} = \mathbb{Q} \cup \{\infty\}$ denotes the set of rational numbers with positive infinity. Summation over $\overline{\mathbb{Q}}$ is done with the convention that $c + \infty = \infty + c = \infty$ for all $c \in \overline{\mathbb{Q}}$.

\subsection{Graphs}
\label{sec:graphs}

 Unless stated otherwise, all graphs are assumed to be undirected and simple. We use $V(G)$ and $E(G)$ to denote the vertex set and edge set of a graph $G$, respectively. We usually write $uv$ instead of $\{u,v\}$ when referring to an edge of a graph. If $X \subseteq V(G)$, then $G[X]$ denotes the subgraph of $G$ induced by $X$. For a positive integer $q$, a \emph{distance-q coloring} of a graph $G$ with $k$ colors is a mapping from $V(G)$ to $[k]$ such that any two distinct vertices $u,v \in V(G)$ connected by a path of length at most $q$ in $G$ are mapped to different colors (i.e. elements of $[k]$). We use $C_{k}$ and $K_k$ to denote, respectively, the cycle graph and the complete graph on $k$ vertices.

In the case of directed graphs, the set of edges $E(G)$ is replaced with a set of \emph{arcs} $A(G) \subseteq V(G)^2$. In particular, we allow loops but all arcs must be distinct. The \emph{underlying graph} of a directed graph $G$ is the graph $G_u$ with vertex set $V(G_u) = V(G)$, such that $vw \in E(G_u)$ if $v \neq w$ and either $(v,w) \in A(G)$ or $(w,v) \in A(G)$.

\paragraph{Minors and topological minors.} Given a graph $G$ and an edge $uv \in E(G)$, the \emph{contraction} of $uv$ consists in merging $u$ and $v$ into a single vertex $w$ whose neighborhood is the union of the neighborhoods of $u$ and $v$ (minus $u$ and $v$, which are removed from $G$). The \emph{deletion} of a vertex $v \in V(G)$ removes $v$ from $V(G)$ along with all edges that contain $v$. A graph $H$ is a \emph{minor} of $G$ if $H$ can be obtained from $G$ through a sequence of edge contractions, vertex deletions, and edge deletions. A \emph{subdivision} of a graph $H$ is a graph obtained from $H$ by replacing edges with pairwise disjoint paths of positive length. A graph $H$ is a \emph{topological minor} of $G$ if there exists a subgraph of $G$ that is a subdivision of $H$.

\paragraph{Tree decompositions.} A \emph{tree decomposition} of a graph $G$ is a pair $(T,\beta)$ such that $T$ is a tree, $\beta$ is a mapping from $V(T)$ to $2^{V(G)}$, for every edge $uv \in E(G)$ there exists $x_{uv} \in V(T)$ such that $uv \subseteq \beta(x_{uv})$, and for every vertex $v \in V(G)$ the set $\{x \in V(T) \mid v \in \beta(x) \}$ induces a non-empty connected subgraph of $T$. The set $\beta(x)$ is called the \emph{bag} of $x \in V(T)$. The \emph{torso} of $x \in V(T)$ is the subgraph of $G$ obtained by adding to $G[\beta(x)]$ all edges $uv$ such that $u,v \in \beta(x) \cap \beta(y)$ for any $xy \in E(T)$. A tree decomposition $(T,\beta)$ of $G$ is \emph{$k$-rich} if for all $xy \in E(T)$, the set $\beta(x) \cap \beta(y)$ induces a clique of size at most $k$ in $G$. Given a subgraph $H$ of $G$, a tree decomposition $(T_H,\beta_H)$ of $H$ is \emph{contained} in $(T,\beta)$ if for all $x \in V(T_H)$, there exists $y \in V(T)$ such that $\beta_H(x) \subseteq \beta(y)$.

\paragraph{Layerings.} A \emph{layering} of a graph $G$ is an ordered partition $L = (V_0,\ldots,V_t)$ of $V(G)$ into \emph{layers} such that for each edge $uv \in E(G)$, the respective layers of $u$ and $v$ are either equal or consecutive in $L$. For $0 \leq i \leq t$, we write $V_{\leq i} = V_0 \cup \ldots \cup V_i$ and $V_{\geq i} = V_i \cup \ldots \cup V_t$. We denote by $G_{\leq i}$, $G_i$, and $G_{\geq i}$ the subgraphs of $G$ respectively induced by $V_{\leq i}$, $V_i$, and $V_{\geq i}$. If $i > 0$, the \emph{shadow} of a connected subgraph $H$ of $G_{\geq i}$ is the set of vertices in $V_{i-1}$ that have at least one neighbor in $V(H)$. A layering of $G$ is \emph{shadow-complete} if every shadow induces a clique in $G$.

\paragraph{Track layouts.} Given a linear ordering $\prec$ of the vertex set of a graph $G$, we say that two edges $u_1u_2,v_1v_2 \in E(G)$ are $\prec$-\emph{crossing} if $u_1 \prec v_1$ and $v_2 \prec u_2$. A \emph{track layout} 
of $G$ is a proper vertex coloring $\gamma$ of $G$ alongside a linear ordering $\prec$ of $V(G)$ such that no two edges $u_1u_2,v_1v_2 \in E(G)$ with $\gamma(u_1) = \gamma(v_1)$ and $\gamma(u_2) = \gamma(v_2)$ are $\prec$-crossing. The color classes of $\gamma$ are called \emph{tracks}. The \emph{track number} of $G$ is the least integer $k$ such that $G$ has a track layout with $k$ tracks.

\paragraph{Homomorphism problems.} A \emph{homomorphism} from a graph $G$ to a graph $H$ is a mapping $g : V(G) \to V(H)$ such that $uv \in E(G) \Rightarrow g(u)g(v) \in E(H)$. For directed graphs, this condition is replaced with $(u,v) \in A(G) \Rightarrow (g(u),g(v)) \in A(H)$. The problem $\plainghom$ takes as input two graphs $G,H$ and asks whether there exists a homomorphism from $G$ to $H$. If $\mathcal H$ is a family of graphs, we write $\ghom{\mathcal H}$ for the restriction of $\plainghom$ to inputs where $H \in \mathcal{H}$. In the case where $\mathcal H$ contains a single graph $H$, we slightly abuse notation and write $\ghom{H}$ instead of $\ghom{\{ H \}}$. Our main result concerns the following problem, which is an optimization variant of $\ghom{\mathcal H}$ for directed graphs.
\problemStatement{
$\valhom{\mathcal H}$
}
{
\begin{minipage}[t]{\textwidth}
Two directed graphs $G,H$ such that the underlying graph of $H$ belongs to $\mathcal H$\\ A mapping $\eta : A(G) \times A(H) \to \overline{\mathbb{Q}}$
\end{minipage}
}
{
Compute the minimum of $\text{cost}(g) = \sum_{(u,v) \in A(G)}\eta((u,v),(g(u),g(v)))$ over all the homomorphisms $g$ from $G$ to $H$, or return $\infty$ if no homomorphism exists.
}
In all theorems involving $\ghom{\mathcal H}$ or $\valhom{\mathcal H}$, the parameter $n$ is the number of vertices in $G$, the parameter $h$ is the number of vertices in $H$, and $|\eta|$ is the encoding size of the mapping $\eta$.

\subsection{Valued constraint satisfaction problems}

Given a finite domain $D$ and a positive integer $k$, a \emph{cost function} over $D$ of arity $k$ is a mapping $D^k \to \overline{\mathbb{Q}}$. A \emph{valued constraint language} over $D$ is a finite set of cost functions over $D$. A valued constraint language is \emph{conservative} if it contains all $\{0,1\}$-valued unary cost functions over its domain. 

A \emph{crisp constraint language} over $D$ is a finite set of relations over $D$. We write symbols for crisp constraint languages in boldface in order to avoid confusion with their valued counterparts. Given two crisp constraint languages $\mathbf{\Gamma},\mathbf{\Delta}$, we write $\mathbf{\Gamma} \leq \mathbf{\Delta}$ if for every $R \in \mathbf{\Gamma}$ there exists $R' \in \mathbf{\Delta}$ such that $R \subseteq R'$. The \emph{feasibility relation} of a cost function $\phi$ of arity $k$ is the $k$-ary relation $\feas{\phi} = \{ t \in D^k \mid \phi(t) < \infty \}$, and the \emph{feasibility language} of a valued constraint language $\Gamma$ is the crisp constraint language $\feas{\Gamma} = \{ \feas{\phi} \mid \phi \in \Gamma\}$. We call a valued constraint language $\Gamma$ a \emph{valuation} of a crisp constraint language $\mathbf{\Gamma}$ if $\feas{\Gamma} \leq \mathbf{\Gamma}$.

\begin{definition}
Let $\Gamma$ be a valued constraint language over a domain $D$. A \emph{\plainVCSP\ instance} $I$ over $\Gamma$ is a finite set of variables $X = \{x_1,\ldots,x_n\}$ together with an objective function of the form
\[
\Phi_I(x_1,\ldots,x_n) = \sum_{i = 1}^m\phi_i(\mathbf{x_i})
\] 
where each $\phi_i$ is a cost function from $\Gamma$, and each $\mathbf{x_i}$ is a tuple of variables from $X$ matching the arity of $\phi_i$.
\end{definition}

The \emph{cost} of an assignment $\alpha : X \to D$ is $\cost{\alpha} = \Phi_I(\alpha(x_1),\ldots,\alpha(x_n))$. We say that $\alpha$ is \emph{feasible} if $\cost{\alpha}$ is finite. The \emph{optimum} $\opt{I}$ of $I$ is the minimum of $\cost{\alpha}$ over all assignments $\alpha$. The problem \VCSP{$\Gamma$} takes as input a \plainVCSP\ instance over $\Gamma$ and asks to compute its optimum. Throughout the paper, we will assume a naive encoding for \plainVCSP\ instances where cost functions are represented as tables matching tuples to their cost.

For any valued constraint language $\Gamma$, \VCSP{$\Gamma$} is either solvable in polynomial time or NP-hard (by results of Kolmogorov, Krokhin, Rol{\'{\i}}nek~\cite{DBLP:journals/siamcomp/KolmogorovKR17}, Bulatov~\cite{DBLP:conf/focs/Bulatov17}, and Zhuk~\cite{DBLP:journals/jacm/Zhuk20}). For our purposes, we will need an algebraic characterization of the crisp constraint languages $\mathbf{\Gamma}$ such that \emph{every} valuation of $\mathbf{\Gamma}$ gives rise to a \plainVCSP\ solvable in polynomial time. We could not find the exact result we require in the literature; however, a simple characterization can be derived from the work of Thapper and \v{Z}ivn\'y~\cite{DBLP:journals/siamcomp/ThapperZ17} on linear relaxations for \plainVCSP s. The remainder of this section presents the concepts from which this characterization (Theorem~\ref{thm:charmaj}) will be constructed.\\

Let $I$ be a \plainVCSP\ instance with domain $D$ and objective function $\Phi_I(x_1,\ldots,x_n)$. Given two integers $0 < k \leq l$, the linear program $\sa{k}{l}{I}$ is defined as follows. For $j \in \mathbb{N}$, let $\phi_{0}^j$ denote the cost function of arity $j$ over $D$ that maps all tuples to $0$. Add to $\Phi_I$ all possible terms of the form $\phi_0^j(\mathbf{x})$ with $j \in [l]$ and $\mathbf{x}$ a tuple of variables of length $j$. Let $\Phi^{l}_I(x_1,\ldots,x_n) = \sum_{i = 1}^m\phi_i(\mathbf{x_i})$ denote the resulting \plainVCSP\ instance. $\sa{k}{l}{I}$ has one variable $\lambda_i(t)$ for all $i \in [m]$ and assignment $t : \text{Set}(\mathbf{x_i}) \to D$, where $\text{Set}(\mathbf{x_i})$ denotes the set of all variables that appear in the tuple $\mathbf{x_i}$.
\begin{align*}
\min \; &\sum_{i= 1}^m \; \; \sum_{\substack{t : \text{Set}(\mathbf{x_i}) \to D\\ \phi_i(t(\mathbf{x_i})) < \infty}} \lambda_i(t) \phi_i(t(\mathbf{x_i}))\\
\text{ subject to } &\sum_{t : \text{Set}(\mathbf{x_i}) \to D} \lambda_i(t) = 1 & \forall i \in [m]\\
                    &\sum_{\substack{t : \text{Set}(\mathbf{x_i}) \to D\\ t_{\restriction{\text{Set}(\mathbf{x_j})}} = s}} \lambda_i(t) = \lambda_j(s) & \forall i,j \in [m], \, s : \text{Set}(\mathbf{x_j}) \to D : \text{Set}(\mathbf{x_j}) \subseteq \text{Set}(\mathbf{x_i}), |\text{Set}(\mathbf{x_j})| \leq k\\
                    &\lambda_i(t) = 0 & \forall i \in [m], \, t : \text{Set}(\mathbf{x_i}) \to D :  \phi_i(t(\mathbf{x_i})) = \infty\\
                    &\lambda_i(t) \geq 0 & \forall i \in [m], \, t : \text{Set}(\mathbf{x_i}) \to D
\end{align*}
Integer solutions to $\sa{k}{l}{I}$ are in one-to-one correspondence with the feasible assignments of $I$, with identical cost. Higher values for $k$ and $l$ provide a hierarchy of increasingly tight linear relaxations of $I$. $\sa{k}{k}{I}$ is often called the \emph{$k$-th level Sherali-Adams relaxation of $I$}. Note that for fixed values of $k$ and $l$, the encoding size of $\sa{k}{l}{I}$ is polynomial in that of $I$.

\begin{definition}
A valued constraint language $\Gamma$ has valued relational width $(k,l)$ if the optimum of $\sa{k}{l}{I}$ is equal to $\opt{I}$ for all instances $I$ of $\VCSP{\Gamma}$.
\end{definition}

Let $k,q$ be two positive integers. If $f$ is an operation $D^k \to D$ and $t_1,\ldots,t_k \in D^q$, we let $f(t_1,\ldots,t_k)$ be the tuple in $D^q$ obtained by componentwise application of $f$ on $(t_1,\ldots,t_k)$:
\[
f(t_1,\ldots,t_k) = ( f(t_1[1], \ldots, t_k[1]), \ldots, f(t_1[q], \ldots, t_k[q]) ).
\]
The operation $f$ is a \emph{polymorphism} of a relation $R$ if $f(t_1,\ldots,t_k) \in R$ for all $t_1,\ldots,t_k \in R$. By extension, $f$ is a polymorphism of a crisp constraint language $\mathbf{\Gamma}$ if $f$ is a polymorphism of every relation in $\mathbf{\Gamma}$. We denote by $\text{Pol}^k(\mathbf{\Gamma})$ the set of all polymorphisms of $\mathbf{\Gamma}$ of arity $k$. 

A $k$-ary \emph{fractional polymorphism} of a valued constraint language $\Gamma$ is a mapping $\omega : \text{Pol}^k(\feas{\Gamma}) \to \mathbb{Q}$ such that $\omega(f) \geq 0$ for all $f \in \text{Pol}^k(\feas{\Gamma})$, $\sum_{f \in \text{Pol}^k(\feas{\Gamma})} \omega(f) = 1$, and
\[
    \sum_{f \in \text{Pol}^k(\feas{\Gamma})} \omega(f)\phi(f(t_1,\ldots,t_k)) \leq \frac{1}{k} \sum_{i = 1}^k \phi(t_i)
\]
for all $\phi \in \Gamma$ and $t_1,\ldots,t_k \in \feas{\phi}$. The \emph{support} of $\omega$ is the set $\supp{\omega} = \{ f \in \text{Pol}^k(\feas{\Gamma}) \mid \omega(f) > 0 \}$, and $\supp{\Gamma}$ is the set of all operations that belong to the support of at least one fractional polymorphism of $\Gamma$. An operation $f : D^3 \to D$ is \emph{majority} if $f(b,a,a) = f(a,b,a) = f(a,a,b) = a$ for all $a,b \in D$.

\begin{theorem}[Thapper and \v{Z}ivn\'y~\cite{DBLP:journals/siamcomp/ThapperZ17}]
\label{thm:vcspwidth}
Let $\Gamma$ be a conservative valued constraint language. If $\supp{\Gamma}$ contains a majority operation then $\Gamma$ has valued relational width $(2,3)$. Otherwise, $\VCSP{\Gamma}$ is NP-hard.
\end{theorem} 

We now present a characterization of crisp constraint languages whose valuations always induce tractable \plainVCSP s. This result is essentially a corollary of Theorem~\ref{thm:vcspwidth}. The general idea is that all valuations of a crisp constraint language $\struct{\Gamma}$ induce polynomial-time \plainVCSP s if and only if their union $\Gamma$ does, and $\supp{\Gamma}$ contains a majority operation if and only if $\struct{\Gamma}$ has a family of polymorphisms that satisfies the three conditions listed in the following definition.

\begin{definition}
\label{def:triple}
Let $\struct{\Gamma}$ be a crisp constraint language over a domain $D$. A \emph{persistent majority support} for $\struct{\Gamma}$ is a tuple $F = (f_1,\ldots,f_{3k})$ of operations from $D^3$ to $D$ such that 
\begin{itemize}
    \item $f_1$ is majority,
    \item for all $a_1,a_2,a_3 \in D$ we have $\{ \{ f_1(a_1,a_2,a_3), \ldots, f_{3k}(a_1,a_2,a_3) \} \} = k \times \{ \{ a_1, a_2, a_3 \} \}$, and
    \item for all $R \in \struct{\Gamma}$ and $t_1,t_2,t_3 \in R$ we have $\{ \{ f_1(t_1,t_2,t_3), \ldots, f_{3k}(t_1,t_2,t_3) \} \} = k \times \{ \{ t_1, t_2, t_3 \} \}$.
\end{itemize}
If $k = 1$ then we call $F$ a \emph{persistent majority triple} for $\struct{\Gamma}$.
\end{definition}

It is easy to see that if $\mathbf{\Gamma}, \mathbf{\Delta}$ are two crisp constraint languages such that $\mathbf{\Gamma} \leq \mathbf{\Delta}$, then every persistent majority support for $\mathbf{\Delta}$ is also a persistent majority support for $\mathbf{\Gamma}$. To simplify notation, we will occasionally interpret persistent majority supports as mappings from $D^3$ to $D^{3k}$ and write $F(a_1,a_2,a_3)$ for the tuple $(f_1(a_1,a_2,a_3),\ldots,f_{3k}(a_1,a_2,a_3))$. In a similar fashion, if $t_1,t_2,t_3$ are tuples of the same arity then we write $F(t_1,t_2,t_3)$ for the tuple $(f_1(t_1,t_2,t_3),\ldots,f_{3k}(t_1,t_2,t_3))$.

\begin{theorem}
\label{thm:charmaj}
Let $\mathbf{\Gamma}$ be a crisp constraint language. Either there exists a valuation $\Gamma$ of $\mathbf{\Gamma}$ such that $\VCSP{\Gamma}$ is NP-hard, or $\mathbf{\Gamma}$ has a persistent majority support and every valuation of $\mathbf{\Gamma}$ has valued relational width $(2,3)$.
\end{theorem}

\begin{proof}
Suppose that $\struct{\Gamma}$ has a persistent majority support $(f_1,\ldots,f_{3k})$ and let $\Gamma$ be a valuation of $\struct{\Gamma}$. Let $\Gamma_c$ be the conservative valued constraint language obtained from $\Gamma$ by addition of all possible $\{0,1\}$-valued unary cost functions over its domain. The second and third items in Definition~\ref{def:triple} ensure that $f_1,\ldots,f_{3k}$ are polymorphisms of $\feas{\Gamma_c}$. Let $m_f$ denote the multiplicity of $f \in \{f_1,\ldots,f_{3k}\}$ in $\{\{ f_1,\ldots,f_{3k} \}\}$, and $\omega : \text{Pol}^3(\feas{\Gamma_c}) \to \mathbb{Q}$ be the mapping that assigns weight $m_f/(3k)$ to each $f \in \{f_1,\ldots,f_{3k}\}$ and $0$ elsewhere. Then, for each function $\phi \in \Gamma_c$ and $t_1,t_2,t_3 \in \feas{\phi}$ we have
\begin{align*}
\sum_{f \in \supp{\omega}} \omega(f) \phi(f(t_1,t_2,t_3)) &= \frac{1}{3k} \left( \sum_{f \in \{f_1,\ldots,f_{3k}\}} m_f \cdot  \phi(f(t_1,t_2,t_3)) \right)\\
&= \frac{1}{3k} \left( \sum_{t \in \{t_1,t_2,t_3\}} \phi(t) \left( \sum_{\substack{f \in \{f_1,\ldots,f_{3k}\}\\ f(t_1,t_2,t_3) = t}} m_f  \right) \right)\\
 &= \frac{1}{3} \left( \sum_{t \in \{t_1,t_2,t_3\}} \phi(t) \right)
\end{align*}
because $\{ \{ f_1(t_1,t_2,t_3), f_2(t_1,t_2,t_3), f_3(t_1,t_2,t_3) \} \} = k \times \{ \{ t_1, t_2, t_3 \} \}$. (This is due to the third item in Definition~\ref{def:triple} if $\feas{\phi}$ is a subset of some relation in $\struct{\Gamma}$, and to the second item otherwise.) Therefore, $\omega$ is a fractional polymorphism of $\Gamma_c$. In addition, $f_1$ is majority so $\supp{\Gamma_c}$ contains a majority operation. By Theorem~\ref{thm:vcspwidth}, we deduce that $\Gamma_c$ has valued relational width $(2,3)$. Since $\Gamma \subseteq \Gamma_c$, $\Gamma$ has valued relational width $(2,3)$ as well.

Conversely, suppose that $\struct{\Gamma}$ does not have a persistent majority support. Let $\Gamma$ be the union of all valuations of $\struct{\Gamma}$ with values in $\{0,1,\infty\}$ and $\Gamma_c$ be the valued constraint language obtained from $\Gamma$ by addition of all possible $\{0,1\}$-valued unary cost functions over its domain. We will prove that $\VCSP{\Gamma_c}$ is NP-hard using Theorem~\ref{thm:vcspwidth} and then deduce that $\VCSP{\Gamma}$ is NP-hard by reduction from $\VCSP{\Gamma_c}$. As $\Gamma$ is a valuation of $\struct{\Gamma}$, the theorem will follow.

For the first step, suppose for the sake of contradiction that $\supp{\Gamma_c}$ contains a majority operation $g_1$. Let $\omega$ be a fractional polymorphism of $\Gamma_c$ with support $\supp{\omega} = \{g_1,\ldots,g_q\}$ and $k, a_1,\ldots,a_q$ be positive integers such that $\omega(g_i) = a_i/(3k)$ for all $1 \leq i \leq q$. Let $F=(f_1,\ldots,f_{3k})$ be the tuple of operations
\[
F = (\underbrace{g_1,\ldots,g_1}_{a_1 \text{ times }},\underbrace{g_2,\ldots,g_2}_{a_2 \text{ times }},\ldots,\underbrace{g_q,\ldots,g_q}_{a_q \text{ times }}).
\]
The operation $f_1 = g_1$ is majority but $F$ cannot be a persistent majority support for $\struct{\Gamma}$. If $F$ violates the third item of Definition~\ref{def:triple}, then there exist $R \in \struct{\Gamma}$, $t_1,t_2,t_3 \in R$ and $j \in \{1,2,3\}$ such that there are fewer occurrences of $t_j$ in $\{\{ f_1(t_1,t_2,t_3), \ldots, f_{3k}(t_1,t_2,t_3) \} \}$ than in $k \times \{ \{ t_1, t_2, t_3 \} \}$. Consider the function $\phi \in \Gamma_c$ such that $\feas{\phi} = R$, $\phi(t_j) = 0$, and $\phi(t) = 1$ for all tuples $t \in \feas{\phi} \setminus \{t_j\}$. Then, we have
\begin{align*}
\sum_{g \in \supp{\omega}} \omega(g) \phi(g(t_1,t_2,t_3)) = \sum_{\substack{i \in [q]\\ g_i(t_1,t_2,t_3) \neq t_j}} \frac{a_i}{3k}
\; &= \; \frac{1}{3k} | \{i \in [3k] \mid f_i(t_1,t_2,t_3) \neq t_j\} |\\
&> \frac{1}{3k} \left(  k \cdot  | \{i \in [3] \mid  t_i \neq t_j\} | \right) = \frac{1}{3} \sum_{i = 1}^3 \phi(t_i)
\end{align*}
which is impossible because $\omega$ is a fractional polymorphism of $\Gamma_c$. If the second item is violated instead, then the same argument using the additional unary cost functions in $\Gamma_c$ instead yields a contradiction as well. We conclude that $\supp{\Gamma_c}$ does not contain a majority operation. By Theorem~\ref{thm:vcspwidth}, $\VCSP{\Gamma_c}$ is NP-hard.

We now sketch a reduction from $\VCSP{\Gamma_c}$ to $\VCSP{\Gamma}$. Let $I$ be an instance of $\VCSP{\Gamma_c}$. Remove all variables that only appear in unary terms in the objective function $\Phi$ and let $b \in \overline{\mathbb{Q}}$ denote the minimum cost of an assignment to these variables. For any other variable $x_i$, there exists a term $\phi_j(\mathbf{x_j})$ in $\Phi$ such that $x_i$ appears in the tuple $\mathbf{x_j}$ and $\phi_j \in \Gamma$. Replace each unary term $\phi_k(x_i)$ such that $\phi_k \notin \Gamma$ with a new term $\phi_j^k(\mathbf{x_j})$, where $\feas{\phi_j^k} = \feas{\phi_j}$ and $\phi_j^k(a(\mathbf{x_j})) = \phi_k(a(x_i))$ for all assignments $a : \text{Set}(\mathbf{x_j}) \to D$ such that $a(\mathbf{x_j}) \in \feas{\phi_j^k}$. The resulting $\plainVCSP$ instance $I'$ is over $\Gamma$ and its optimum satisfies $\opt{I} = \opt{I'}+b$, where $b$ can be computed in linear time.
\end{proof}

We remark without proof that there exists a universal constant $c$ such that a crisp constraint language has a persistent majority support if and only if it has one composed of at most $c$ operations. We suspect that three operations always suffice, although we could not prove this particular bound. Thus, the problem of deciding whether a crisp constraint language satisfies the criterion of Theorem~\ref{thm:charmaj} is in NP. We finish this section with a straightforward lemma that will allow us to ignore this problem for the rest of the paper.

\begin{lemma}
\label{lem:unif}
There exists a polynomial-time algorithm that takes as input a valued constraint language $\Gamma$ and a \plainVCSP\ instance $I$ over $\Gamma$, and either outputs an assignment of minimum cost to $I$ or concludes that $\feas{\Gamma}$ does not have a persistent majority support.
\end{lemma}

\begin{proof}
We use a standard self-reduction argument. Let $D$ be the domain of $\Gamma$. For each $a \in D$, let $\phi_a$ be the unary cost function over $D$ such that $\phi_a(a) = 0$ and $\phi_a(b) = \infty$ for all $b \in D \setminus \{a\}$. Let $\Gamma^* = \Gamma \cup \{ \phi_a \mid a \in D \}$. We say that a variable $x$ in a \plainVCSP\ instance over $\feas{\Gamma^*}$ is \emph{fixed} if the objective function contains a term $\phi_a(x)$ for some $a \in D$.

The algorithm starts by solving $\sa{2}{3}{I}$. If $\sa{2}{3}{I}$ is infeasible, then all assignments to the variables in $I$ have infinite cost. In this case, we return an arbitrary assignment. Otherwise, let $s$ denote the optimum of $\sa{2}{3}{I}$. Repeat the following procedure until all variables in $I$ are fixed. Pick a variable $x_i$ that is not fixed. For each possible $a \in D$, construct a \plainVCSP\ instance $I_{x_i \gets a}$ by adding to the objective function of $I$ the term $\phi_{a}(x_i)$. Solve all linear programs $\sa{2}{3}{I_{x_i \gets a}}$. If all optimum values are strictly greater than $s$, return that $\feas{\Gamma}$ does not have a persistent majority support. Otherwise, pick $a$ such that the optimum of $\sa{2}{3}{I_{x_i \gets a}}$ is exactly $s$ and set $I \vcentcolon= I_{x_i \gets a}$. If all variables in $I$ are fixed, then return the assignment that maps each variable $x_i$ to the (unique) $a \in D$ such that the objective function of $I$ contains a term $\phi_a(x_i)$.

Whenever the algorithm above returns an assignment, its cost matches the optimum $s$ of $\sa{2}{3}{I}$ (or positive infinity if $\sa{2}{3}{I}$ is not feasible). Since $s \leq \opt{I}$, this assignment must be one of minimum cost. On the other hand, if the algorithm returns that $\feas{\Gamma}$ does not have a persistent majority support, then there exists a \plainVCSP\ instance $I^*$ over $\Gamma^*$ such that $\sa{2}{3}{I^*}$ is feasible with optimum value $s$ but no variable assignment has cost $s$. This implies that $\Gamma^*$ does not have valued relational width $(2,3)$. Then, by Theorem~\ref{thm:charmaj}, $\feas{\Gamma^*}$ does not have a persistent majority support. Since $\feas{\Gamma^*}$ has a persistent majority support if and only if $\feas{\Gamma}$ does, the output of the algorithm is correct.
\end{proof}

\section{Persistent majority colorings}
\label{sec:dmn}

Consider an input $(G,H,\eta)$ to $\valhom{\mathcal H}$ and two fixed mappings $\gamma_G : V(G) \to [k]$, $\gamma_H : V(H) \to [k]$. The problem of computing the minimum cost of a homomorphism $g$ from $G$ to $H$ such that $\gamma_H(g(v)) = \gamma_G(v)$ for all $v \in V(G)$ has a direct formulation as a \plainVCSP\ instance with one term per arc of $G$. Its constraint language may or may not be tractable; in general, this depends on the mappings and the cost function $\eta$. In this section, we describe a structural property of $\gamma_H$ (based on Theorem~\ref{thm:charmaj}) that guarantees tractability independently of $\gamma_G$ and $\eta$.

Let $G$ be a graph and $\gamma : V(G) \to [k]$ be a proper vertex coloring of $G$. For each pair of colors $(i,j) \in [k]^2$, define $R_{ij}$ as the binary relation $R_{ij} = \{ (u,v) \in V(G)^2 \mid (\gamma(u) = i) \land (\gamma(v) = j) \land (uv \in E(G)) \}$. Let $\struct{\Gamma_{G,\gamma}}$ denote the crisp constraint language $\{ R_{ij} \mid (i,j) \in [k]^2\}$. 

\begin{definition}
\label{def:col}
    A \emph{persistent majority coloring} of a graph $G$ is a proper vertex coloring $\gamma$ of $G$ such that $\struct{\Gamma_{G,\gamma}}$ has a persistent majority triple.
\end{definition}

The \emph{persistent majority number} $\majcover{G}$ of a graph $G$ is the least integer $k$ such that $G$ has a persistent majority coloring with $k$ colors.

\begin{lemma}
\label{lem:algo}
If ${\mathcal H}$ is a family of graphs of bounded persistent majority number, then $\valhom{\mathcal H}$ can be solved in time $2^{O(n+h)}|\eta|^{O(1)}$ and polynomial space. In addition, if there exists a polynomial-time algorithm that takes as input a graph $H \in \mathcal H$ and outputs a persistent majority coloring of $H$ with $O(1)$ colors, then $\valhom{\mathcal H}$ can be solved in time $2^{O(n)} (h + |\eta|)^{O(1)}$ and polynomial space.
\end{lemma}

\begin{proof}
Let $(G,H,\eta)$ be an instance of $\valhom{\mathcal H}$ and $H_u$ be the underlying graph of $H$. Let $V(G) = (v_i)_{i \in [n]}$. We suppose (without loss of generality) that each vertex $v_i \in V(G)$ is incident to at least one arc. Starting from $k = 1$, enumerate all pairs $(\gamma_G,\gamma_H)$ such that $\gamma_G : V(G) \to [k]$ is any mapping and $\gamma_H : V(H) \to [k]$ is a proper vertex coloring of $H_u$. Call a homomorphism $g$ from $G$ to $H$ \emph{compliant} with a pair $(\gamma_G,\gamma_H)$ if $\gamma_H(g(v_i)) = \gamma_G(v_i)$ for all $v_i \in V(G)$.
For each pair $(\gamma_G,\gamma_H)$, construct a \plainVCSP\ instance as follows. The domain is $V(H)$ and the set of variables $X$ contains one variable $x_i$ for each $v_i \in V(G)$. The objective function is
$\Phi_{(\gamma_G,\gamma_H)}(x_1,\ldots,x_n) = \sum_{(v_i,v_j) \in A(G)} \phi_{ij}(x_i,x_j)$,
where 
\[
\phi_{ij}(u,w) =
\begin{cases}
\eta((v_i,v_j),(u,w)) &\text{ if } (u,w) \in A(H) \text{ and } (\gamma_G(v_i),\gamma_G(v_j)) = (\gamma_H(u),\gamma_H(w))\\
\infty &\text{ otherwise.}
\end{cases}
\]
The feasible assignments of $\Phi_{(\gamma_G,\gamma_H)}$ are in one-to-one correspondence with the homomorphisms from $G$ to $H$ of finite cost that are compliant with $(\gamma_G,\gamma_H)$. Now, run the algorithm $A$ of Lemma~\ref{lem:unif} on $\Phi_{(\gamma_G,\gamma_H)}$ for each pair $(\gamma_G,\gamma_H)$. If there exists a coloring $\gamma_H^* : V(H) \to [k]$ such that $A$ returns an assignment for all pairs where $\gamma_H = \gamma_H^*$, then stop and return the minimum cost over all these assignments. Otherwise, increment $k$ and repeat the procedure above. This completes the description of the algorithm.

For each homomorphism $g$ from $G$ to $H$ and coloring $\gamma_H : V(H) \to [k]$, there exists a mapping $\gamma_G : V(G) \to [k]$ such that $g$ is compliant with $(\gamma_G,\gamma_H)$. (Simply define $\gamma_G(v_i)$ as $\gamma_H(g(v_i))$.) In addition, any assignment returned by $A$ has optimum cost for the input \plainVCSP\ instance. Thus, when the algorithm above stops it returns the optimum of $(G,H,\eta)$.

For termination, let $\gamma_H^*$ be a persistent majority coloring of $H_u$ with $\majcover{H_u}$ colors. For all pairs $(\gamma_G,\gamma_H^*)$, the \plainVCSP\ instance $\Phi_{(\gamma_G,\gamma_H^*)}$ is over a language that is a valuation of $\struct{\Gamma_{H_u,\gamma_H^*}} \cup \{ = \}$. (If $\gamma_G(v_i) \neq \gamma_G(v_j)$, then $\feas{\phi_{ij}}$ is composed only of arcs $(u,w) \in A(H)$ with $u \neq w$, $\gamma_H^*(u) = \gamma_G(v_i)$, and $\gamma_H^*(w) = \gamma_G(v_j)$. Therefore, $\feas{\phi_{ij}}$ is a subset of $R_{\gamma_G(v_i)\gamma_G(v_j)} \in \struct{\Gamma_{H_u,\gamma_H^*}}$. If $\gamma_G(v_i) = \gamma_G(v_j)$ instead, then each tuple in $\feas{\phi_{ij}}$ is a loop in $H$ so $\feas{\phi_{ij}}$ is a subset of the equality relation over $V(H)$.) If $\struct{\Gamma_{H_u,\gamma_H^*}}$ has a persistent majority triple then so does $\struct{\Gamma_{H_u,\gamma_H^*}} \cup \{ = \}$, hence by Theorem~\ref{thm:charmaj} the algorithm $A$ will return an assignment for all these instances. This implies that the algorithm will stop at a certain value $k \leq \majcover{H_u}$. Each call to $A$ takes time polynomial in $n, h, |\eta|$ (including the time to construct the \plainVCSP\ instance), and for any $k > 0$ there exist $O(k^{n+h})$ possible pairs of mappings $(\gamma_G,\gamma_H)$. The complexity of the whole algorithm is therefore in $O(\majcover{H_u}^{n+h} \cdot \text{poly}(n, h, |\eta|)) = 2^{O(n+h)} |\eta|^{O(1)}$, or $2^{O(n)} (h + |\eta|)^{O(1)}$ if an appropriate persistent majority coloring of $H_u$ can be computed in polynomial time.
\end{proof}

We will use Lemma~\ref{lem:algo} to identify new families of graphs $\mathcal H$ for which $\valhom{\mathcal H}$ (and hence $\ghom{\mathcal H}$) can be solved in single-exponential time. We start by presenting a few lemmas that establish elementary facts about persistent majority numbers and relate them to more standard graph measures.

\begin{lemma}
    \label{lem:wellbehaved}
    Let $G,H$ be graphs and $k$ be a positive integer. The following statements are true:
    \begin{itemize}
        \item If $H$ is a subgraph of $G$, then $\majcover{H} \leq \majcover{G}$.
        \item If $H$ can be obtained from $G$ by deleting at most $k$ vertices, then $\majcover{G} \leq \majcover{H} + k$.
    \end{itemize}
\end{lemma}

\begin{proof}
    If $H$ is a subgraph of $G$ and $\gamma$ is a vertex coloring of $G$, then $\struct{\Gamma_{H,\gamma}} \leq \struct{\Gamma_{G,\gamma}}$. Therefore, any persistent majority triple for $\struct{\Gamma_{G,\gamma}}$ is also a persistent majority triple for $\struct{\Gamma_{H,\gamma}}$. This establishes the first item of the lemma. For the second item, suppose that $H$ can be obtained from $G$ by deleting at most $k$ vertices. Let $\gamma : V(H) \to [\majcover{H}]$ be a persistent majority coloring of $H$ and $F=(f_1,f_2,f_3)$ be a persistent majority triple for $\struct{\Gamma_{H,\gamma}}$. Extend $\gamma$ to a coloring of $G$ by giving each deleted vertex $v \in V(G) \setminus V(H)$ its own unique color. Let $F'$ be a triple of operations $V(G)^3 \to V(G)$ such that 
    \[
    F'(u,v,w) = 
    \begin{cases}
        F(u,v,w) &\text{ if } (u,v,w) \in V(H)^3\\
        (u,v,w) &\text{ if } (u,v,w) \notin V(H)^3 \text{ and } v \neq w\\
        (w,v,u) &\text{ if } (u,v,w) \notin V(H)^3 \text{ and } v = w\\
    \end{cases}
    \]
    It is straightforward to check that the first two conditions of Definition~\ref{def:triple} are satisfied. In addition, any relation $R_{ij} \in \struct{\Gamma_{G,\gamma}} \setminus \struct{\Gamma_{H,\gamma}}$
    is such that all tuples $t \in R_{ij}$ coincide on either their first or second entry. For relations of this kind, any triple of operations that satisfies the second condition of Definition~\ref{def:triple} also satisfies the third. The extended coloring $\gamma$ uses $\majcover{H} + k$ colors in total. 
\end{proof}

\begin{lemma}
\label{lem:stack}
The persistent majority number of a graph is at most its track number.
\end{lemma}

\begin{proof}
Let $G$ be a graph and $(\gamma,\prec)$ be a track layout of $G$ with $k$ tracks. Let $F=(f_1,f_2,f_3)$ be the triple of operations $V(G)^3 \to V(G)$ given by 
\begin{align*}
f_1(u,v,w) &= \text{median}_{\prec}(u,v,w),\\
f_2(u,v,w) &= \text{min}_{\prec}(u,v,w),\\
f_3(u,v,w) &= \text{max}_{\prec}(u,v,w).
\end{align*}
By definition, $f_1$ is majority. We will use the observation that for any set $S$ equipped with a linear order and $a,b,c \in S$, the triple $(\text{median}(a,b,c), \text{min}(a,b,c), \text{max}(a,b,c))$ is always a permutation of $(a,b,c)$. This implies in particular that $F$ satisfies the second item in Definition~\ref{def:triple}, leaving only the third to be verified.

Let $R_{ij} \in \struct{\Gamma_{G,\gamma}}$ and note that $\prec$ induces a linear order $\prec_{ij}$ on the tuples of $R_{ij}$, where $(v_1,v_2) \prec_{ij} (v_3,v_4)$ if either $v_1 \prec v_3$ or $v_2 \prec v_4$. (Totality is immediate. If $(v_1,v_2),(v_3,v_4) \in R_{ij}$ are two tuples such that $(v_1,v_2) \prec_{ij} (v_3,v_4)$, then $v_1 \preceq v_3$ and $v_2 \preceq v_4$ since otherwise the edges $v_1v_2$ and $v_3v_4$ would be $\prec$-crossing. This ensures both antisymmetry and transitivity.) Now, if $(v_1,v_2), (v_3,v_4), (v_5,v_6) \in R_{ij}$ and $(v_p,v_q) = \text{min}_{\prec_{ij}}((v_1,v_2), (v_3,v_4), (v_5,v_6))$, then $v_p = \text{min}_{\prec}(v_1,v_3,v_5)$ and $v_q = \text{min}_{\prec}(v_2,v_4,v_6)$. Therefore, we have $f_2((v_1,v_2), (v_3,v_4), (v_5,v_6)) = \text{min}_{\prec_{ij}}((v_1,v_2), (v_3,v_4), (v_5,v_6))$. With a similar argument, we deduce $f_3((v_1,v_2), (v_3,v_4), (v_5,v_6)) = \text{max}_{\prec_{ij}}((v_1,v_2), (v_3,v_4), (v_5,v_6))$ and $f_1((v_1,v_2), (v_3,v_4), (v_5,v_6)) = \text{median}_{\prec_{ij}}((v_1,v_2), (v_3,v_4), (v_5,v_6))$. By the observation of the first paragraph applied to $\prec_{ij}$ and $S = R_{ij}$, it follows that $F(t_1,t_2,t_3)$ is a permutation of $(t_1,t_2,t_3)$ for all $t_1,t_2,t_3 \in R_{ij}$. This is true for all choices of $i,j$ so $F$ is a persistent majority triple for $\struct{\Gamma_{G,\gamma}}$. As $\gamma$ uses $k$ colors, the lemma follows.
\end{proof}

In order to illustrate the relationship between track layouts and tractable constraint satisfaction problems, consider an input $(G,H)$ to \plainghom. Fix a track layout of $H$, and suppose that we have already decided to which track each vertex of $G$ will be mapped. Then, each vertex $u$ of $G$ has a list of candidate vertices $L(u)$, all belonging to the same track of $H$. If $uv$ is an edge of $G$ and some vertex $w  \in L(u)$ has no neighbors in $L(v)$, then $u$ cannot be mapped to $w$ as this would leave no option for mapping $v$. Thus, $w$ can safely be removed from $L(u)$. Repeating this process will gradually shrink the lists of each vertex. If a list becomes empty, then there is no homomorphism from $G$ to $H$ that respects this particular track assignment. On the other hand, if the reduction rule cannot be applied further and all lists are non-empty, then for each edge $uv$ in $G$ the maximum vertices (with respect to the layout ordering $\prec$) in the respective lists of $u$ and $v$ must be neighbors in $H$. This is because each of these maximum vertices has at least one neighbor in the other list and edges cannot cross. Therefore, the mapping that sends each $u \in V(G)$ to the maximum vertex that remains in $L(u)$ is a homomorphism from $G$ to $H$.
 
In standard \plainCSP\ terminology, the property obtained after repeated application of this reduction rule is called \emph{arc consistency}. The constraints of the linear program $\sa{2}{3}{\cdot}$ implicitly enforce arc consistency, in the sense that a feasible solution cannot give nonzero weight to a variable assignment that could have been ruled out by arc consistency.

The next theorem is a direct corollary of Lemma~\ref{lem:algo} and Lemma~\ref{lem:stack}.

\begin{theorem}
    \label{thm:trackhom}
    If $\mathcal H$ is a family of graphs of bounded track number, then $\valhom{\mathcal H}$ can be solved in time $2^{O(n+h)} |\eta|^{O(1)}$ and polynomial space.
\end{theorem}

Dujmovi\'c and Wood~\cite{DBLP:journals/dmtcs/DujmovicW05} proved that every graph has a subdivision whose track number is at most $4$. This implies that the family of all graphs of track number at most $4$ does not exclude a topological minor, and hence Theorem~\ref{thm:trackhom} is incomparable with Theorem~\ref{thm:valhom} (which we will prove in the next section).

The connection between track layouts and persistent majority colorings has further algorithmic implications. Dujmovi\'c et al.~\cite{DBLP:journals/jacm/DujmovicJMMUW20} proved that all proper minor-closed graph families have bounded track number. Crucially, their proof is constructive and provides a polynomial-time algorithm for computing a layout with a bounded number of tracks. This algorithm, combined with Lemma~\ref{lem:stack}, will constitute a key ingredient in the proof of Theorem~\ref{thm:valhom}.

\begin{lemma}
\label{lem:minor}
For every graph $H$, there exists a constant $k=k(H)$ such that every graph $G$ excluding $H$ as a minor has persistent majority number at most $k$. In addition, there exists a polynomial-time algorithm that takes as input a graph $G$ and either outputs a persistent majority coloring of $G$ with $O(1)$ colors or concludes that $G$ contains $H$ as a minor.
\end{lemma}

The final lemma of this section establishes that graphs of bounded degree have bounded persistent majority number.

\begin{lemma}
\label{lem:degree}
The persistent majority number of a graph of maximum degree $k$ is at most $k^2+1$.
\end{lemma}

\begin{proof}
Let $G$ be a graph of maximum degree $k$. Each vertex in $G$ has at most $k$ neighbors and $k(k-1)$ distance-$2$ neighbors, so there exists a proper distance-$2$ coloring $\gamma$ of $G$ with $c \leq k^2+1$ colors. A straightforward greedy algorithm can find $\gamma$ in polynomial time. Let $(f_1,f_2,f_3)$ be the triple of operations $V(G)^3 \to V(G)$ given by
\[
f_1(u,v,w) = 
\begin{cases}
    u \text{ if } v \neq w\\
    w \text{ otherwise}
\end{cases}
\hspace{15mm}
f_2(u,v,w) = 
\begin{cases}
    v \text{ if } u \neq w\\
    w \text{ otherwise}
\end{cases}
\]
\[
f_3(u,v,w) = 
\begin{cases}
    \text{minority}(u,v,w) \text{ if } |\{u,v,w\}| = 2\\
    w \text{ otherwise}
\end{cases}
\]
where $\text{minority}(u,v,w)$ denotes the element with lowest nonzero multiplicity in $\{\{u,v,w\}\}$. This specific triple of operations was introduced by Cohen et al.~\cite{DBLP:journals/ai/CohenCJK06}. By definition, both $f_1$ and $f_2$ are majority.

We start by proving the second item in Definition~\ref{def:triple}. Let $u,v,w$ be three vertices of $G$. If all three vertices are distinct, then $f_1(u,v,w) = u$, $f_2(u,v,w) = v$ and $f_3(u,v,w) = w$. If exactly two vertices among $u,v,w$ are equal, then $f_1$ and $f_2$ return these two vertices (because they are majority operations) and $f_3$ returns the third. If $(u,v,w) = (u,u,u)$, then all three operations return $u$. In all cases, we have $\{\{f_1(u,v,w), f_2(u,v,w), f_3(u,v,w)\}\} = \{\{u,v,w\}\}$.

Now, let $R_{ij} \in \struct{\Gamma_{G,\gamma}}$ and $t_1,t_2,t_3$ be three tuples of $R_{ij}$. Since $\gamma$ is a distance-2 coloring of $G$, no vertex of color $i$ has two distinct neighbors of color $j$ (and vice versa). Therefore, distinct tuples in $R_{ij}$ have pairwise disjoint entries. This means we can repeat the same reasoning as above: if all three tuples are distinct (or if they are all equal) then $f_i(t_1,t_2,t_3) = t_i$ for all $i \in \{1,2,3\}$, and if any two tuples are equal then $f_1$ and $f_2$ return these two tuples whereas $f_3$ returns the third. In all cases it holds that $\{\{f_1(t_1,t_2,t_3), f_2(t_1,t_2,t_3), f_3(t_1,t_2,t_3)\}\} = \{\{t_1,t_2,t_3\}\}$, so $(f_1,f_2,f_3)$ is a majority triple for $\struct{\Gamma_{G,\gamma}}$. The lemma then follows from the bound $c \leq k^2+1$.
\end{proof}

\section{Excluding a topological minor}

This section contains the proof of Theorem~\ref{thm:valhom}. Our starting point is the following theorem of Grohe and Marx, which is a topological analog of the graph minor structure theorem of Robertson and Seymour~\cite{DBLP:journals/jct/RobertsonS03a}.

\begin{theorem}[Grohe and Marx~\cite{DBLP:journals/siamcomp/GroheM15}]
\label{thm:structuretop}
For every graph $H$, there exists a constant $k_1 = k_1(H)$ such that every graph excluding $H$ as a topological minor has a tree decomposition whose torsos either exclude $K_{k_1}$ as a minor or contain at most $k_1$ vertices of degree greater than $k_1$. Moreover, there exists a polynomial-time algorithm that takes as input a graph $G$ and outputs either such a decomposition or a subdivision of $H$ in $G$.
\end{theorem}

By our observations from the previous section (Lemma~\ref{lem:wellbehaved}, Lemma~\ref{lem:minor}, and Lemma~\ref{lem:degree}), each torso in the tree decomposition described in Theorem~\ref{thm:structuretop} has an efficiently computable persistent majority coloring with a constant number of colors. We show that these ``local'' colorings can be combined into a persistent majority coloring of the whole graph without increasing too much the number of colors. Our proof strategy closely follows an argument of Dujmovi\'c, Morin, and Wood~\cite{DBLP:journals/jct/DujmovicMW17}, who used it to prove that the track number of proper minor-closed families of graphs is at most polylogarithmic in the number of vertices.

We first borrow from them a couple of lemmas that connect Theorem~\ref{thm:structuretop} to shadow-complete layerings. The first is a straightforward adaptation of~\cite[Lemma~26]{DBLP:journals/jct/DujmovicMW17}. The original statement concerns graphs with excluded minors, rather than topological minors; however, the proof is essentially identical. 

\begin{lemma}[adapted from Dujmovi\'c, Morin and Wood~\cite{DBLP:journals/jct/DujmovicMW17}]
\label{lem:rich}
For every graph $H$, there exists a constant $k_2 = k_2(H)$ such that every graph $G$ excluding $H$ as a topological minor is a subgraph of a graph $G'$ that has a $k_2$-rich tree decomposition whose bags either exclude $K_{k_2}$ as a minor or contain at most $k_2$ vertices of degree greater than $k_2$. Moreover, there exists a polynomial-time algorithm that takes as input a graph $G$ and outputs either $G'$ and its tree decomposition or a subdivision of $H$ in $G$.
\end{lemma}

\begin{proof}
Let $(T,\beta)$ be a tree decomposition of $G$ as in Theorem~\ref{thm:structuretop}. For each edge $xy \in E(T)$, add to $G$ all possible edges $uv \subseteq \beta(x) \cap \beta(y)$. Let $G'$ be the graph resulting from this process and $(T',\beta')$ be the tree decomposition of $G'$ with $T' = T$ and $\beta' = \beta$. The subgraphs of $G'$ induced by the bags of $(T',\beta')$ correspond to the torsos of $(T,\beta)$ in $G$, so each bag induces a subgraph of $G'$ that either excludes $K_{k_1}$ as a minor or contains at most $k_1$ vertices of degree greater than $k_1$. In particular, these subgraphs cannot contain a clique of size $k_1 + 2$. Since $\beta'(x) \cap \beta'(y)$ induces a clique of $G'$ for all $xy \in E(T')$, $(T',\beta')$ is a $k_2$-rich tree decomposition of $G'$ with $k_2 = k_1 + 2$.
\end{proof}

The second lemma corresponds to~\cite[Lemma~27]{DBLP:journals/jct/DujmovicMW17}, with no changes besides an algorithmic claim that follows immediately from the original proof.

\begin{lemma}[Dujmovi\'c, Morin and Wood~\cite{DBLP:journals/jct/DujmovicMW17}]
\label{lem:layers}
If a graph $G$ has a $k$-rich tree decomposition $(T,\beta)$, then $G$ has a shadow-complete layering $(V_0,\ldots,V_t)$ such that every shadow has size at most $k$ and each subgraph $G_i$ has a $(k-1)$-rich tree decomposition contained in $(T,\beta)$. Given $G$ and $(T,\beta)$ in input, the layering $(V_0,\ldots,V_t)$ and the tree decompositions of each $G_i$ can be computed in polynomial time.
\end{lemma}

Now, we will prove that the persistent majority number of a graph with a shadow-complete layering is bounded by a function that only depends on the maximum shadow size and the maximum persistent majority number of its layers. Starting from the $k_2$-rich tree decomposition of Lemma~\ref{lem:rich} and iterating Lemma~\ref{lem:layers} will then yield the desired result.

For the next couple of lemmas, it will be convenient to use the following alternative definition of persistent majority triples for binary crisp constraint languages. Call an operation $\pi : D^3 \to D^3$ a \emph{coordinate permutation} if there exists a bijective mapping $\sigma : \{1,2,3\} \to \{1,2,3\}$ such that $\pi(a_1,a_2,a_3) = (a_{\sigma(1)}, a_{\sigma(2)}, a_{\sigma(3)})$ for all $(a_1,a_2,a_3) \in D^3$. If $\struct{\Gamma}$ is a binary crisp constraint language over $D$ and $F = (f_1,f_2,f_3)$ is a triple of operations from $D^3$ to $D$, then $F$ is a persistent majority triple for $\struct{\Gamma}$ if and only if 
\begin{itemize}
    \item[(${\dagger}$)] $f_1$ is majority,
    \item[(${\dagger}{\dagger}$)] for all $a_1,a_2,a_3 \in D$, there exists a coordinate permutation $\pi$ such that $F(a_1,a_2,a_3) = \pi(a_1,a_2,a_3)$, and
    \item[(${\dagger}{\dagger}{\dagger}$)] for all $R \in \struct{\Gamma}$ and $(a_1,b_1),(a_2,b_2),(a_3,b_3) \in R$, there exists a coordinate permutation $\pi$ such that $F(a_1,a_2,a_3) = \pi(a_1,a_2,a_3)$ and $F(b_1,b_2,b_3) = \pi(b_1,b_2,b_3)$.
\end{itemize}

\begin{lemma}
\label{lem:permu-out}
Let $G$ be a graph, $\gamma$ be a persistent majority coloring of $G$, and $F$ be a persistent majority triple for $\struct{\Gamma}_{G,{\gamma}}$. If $\kappa_1,\kappa_2,\kappa_3$ are the vertex sets of three cliques of $G$, then there exists a coordinate permutation $\pi_{\kappa_1\kappa_2\kappa_3}$ such that $F(u,v,w) = \pi_{\kappa_1\kappa_2\kappa_3}(u,v,w)$ for all monochromatic triples $(u,v,w) \in \kappa_1 \times \kappa_2 \times \kappa_3$.
\end{lemma}

\begin{proof}
Each clique of $G$ contains at most one vertex of each color, so $\kappa_1 \times \kappa_2 \times \kappa_3$ contains at most one monochromatic triple of each color. First, suppose that there exists a color $p$ whose associated triple $(u_1,u_2,u_3) \in \kappa_1 \times \kappa_2 \times \kappa_3$ contains three distinct vertices. Define $\pi_{\kappa_1\kappa_2\kappa_3}$ as the unique coordinate permutation such that $F(u_1,u_2,u_3) = \pi_{\kappa_1\kappa_2\kappa_3}(u_1,u_2,u_3)$. For any monochromatic triple $(v_1,v_2,v_3) \in \kappa_1 \times \kappa_2 \times \kappa_3$ with $(v_1,v_2,v_3) \neq (u_1,u_2,u_3)$, the vertices $v_1,v_2,v_3$ have color $p' \neq p$. The relation $R_{pp'} \in \struct{\Gamma}_{G,{\gamma}}$ contains the three tuples $(u_1,v_1), (u_2,v_2), (u_3,v_3)$ because each $\kappa_i$ induces a clique of $G$. Then, by condition (${\dagger}{\dagger}{\dagger}$) in the (alternative) definition of persistent majority triples we must have $F(v_1,v_2,v_3) = \pi_{\kappa_1\kappa_2\kappa_3}(v_1,v_2,v_3)$.

For the rest of the proof, suppose that every monochromatic triple in $\kappa_1 \times \kappa_2 \times \kappa_3$ contains at most two distinct vertices. We ignore triples where the same vertex occurs thrice because they are invariant under every coordinate permutation. For any monochromatic triple $m$, let $\alpha(m)$ denote the index $i  \in \{1,2,3\}$ such that the $i$th entry of $m$ is the vertex that appears exactly once in $m$. Similarly, let $\beta(m)$ denote the index of that vertex in $F(m)$. If there exists a bijective mapping $\sigma : \{1,2,3\} \to \{1,2,3\}$ such that $\beta(m) = \sigma(\alpha(m))$ for all monochromatic triples $m$, then the coordinate permutation $\pi_{\kappa_1\kappa_2\kappa_3}(a_1,a_2,a_3) = (a_{\sigma(1)}, a_{\sigma(2)}, a_{\sigma(3)})$ satisfies the statement of the lemma. In turn, it is sufficient to show that for any two monochromatic triples $m_1,m_2$, we have $\beta(m_1) = \beta(m_2)$ if and only if $\alpha(m_1) = \alpha(m_2)$.

Let $p_1$ and $p_2$ be the colors associated with $m_1$ and $m_2$, respectively. If $\alpha(m_1) = \alpha(m_2)$, then $R_{p_1 p_2}$ contains the tuples $t_1 = (u_1,u_2)$ and $t_2 = (v_1,v_2)$, where $u_i$ is the vertex that appears twice in $m_i$ and $v_i$ is the vertex that appears once. If $\beta(m_1) \neq \beta(m_2)$, then one triple among $F(t_1,t_1,t_2), F(t_1,t_2,t_1), F(t_2,t_1,t_1)$ (the one with $t_2$ at position $\alpha(m_1)$) does not contain $t_2$, which is impossible because $F$ is a persistent majority triple for $\struct{\Gamma}_{G,{\gamma}}$. Hence, $\alpha(m_1) = \alpha(m_2)$ implies $\beta(m_1) = \beta(m_2)$. Conversely, if $\alpha(m_1) \neq \alpha(m_2)$ then $t_1 = (u_1,u_2)$, $t_2 = (v_1,u_2)$, and $t_3 = (u_1,v_2)$ belong to $R_{p_1 p_2}$. If $\beta(m_1) = \beta(m_2)$, then there exists a coordinate permutation $\pi$ such that $(v_1,v_2)$ appears as the $\beta(m_1)$th entry of $F(\pi(t_1,t_2,t_3))$ despite not being present among $t_1,t_2,t_3$, a final contradiction.
\end{proof}

\begin{lemma}
\label{lem:shadowcombine}
If a graph G has a shadow-complete layering $(V_0,\ldots,V_t)$ such that each layer induces a subgraph with persistent majority number at most $c$ and each shadow has size at most $s$, then G has persistent majority number at most $3(c+1)^{s+1}$. In addition, a persistent majority coloring of $G$ matching that bound can be computed in polynomial time if a persistent majority coloring of each layer is provided in input.
\end{lemma}

\begin{proof}
We assume without loss of generality that $G$ is connected and $V_0 \neq \emptyset$. For each $0 \leq i \leq t$, let $\comp{G_i}$ denote the set of connected components of $G_i$, and define ${\mathcal C} = \cup_i \comp{G_i}$. For $i > 0$ and $v \in V_i$, define $\kappa(v)$ as the shadow of the (unique) graph $X \in \mathcal C$ whose vertex set contains $v$. The layering $(V_0,\ldots,V_t)$ is shadow-complete, so $\kappa(v)$ induces a clique of $G$ that we call the \emph{parent clique} of $v$. By convention, we set $\kappa(v) = \emptyset$ for all $v \in V_0$. We will use the following observation.

\renewcommand\qedsymbol{$\blacksquare$}
\begin{claim}
    \label{claim:g0}
    $G_0$ is connected and $\kappa(v) \neq \emptyset$ for all $v \in V_{\geq 1}$.
\end{claim}

\begin{proof}
Let $v_1,v_2$ be two vertices that belong to the same layer $V_i$, for some $0 \leq i \leq t$. We show that $v_1$ and $v_2$ belong to the same connected component of $G_i$ if and only if they belong to the same connected component of $G_{\geq i}$. This will establish both parts of the claim: $G_{\geq 0} = G$ is connected so $G_0$ must be connected as well, and $\kappa(v) = \emptyset$ implies that the connected component of $v \in V_i$ in $G_{\geq i}$ is also a connected component of $G$. Since $G$ is connected and $V_0$ is not empty, this is only possible if $v \in V_0$.

The forward implication is immediate, so let us assume that $v_1$ and $v_2$ belong to the same connected component of $G_{\geq i}$. Let $P = (u_1,\ldots,u_m)$ be a shortest path from $u_1 = v_1$ to $u_m = v_2$ in $G_{\geq i}$. If $u_j \in V_i$ for all $j$, then $v_1$ and $v_2$ belong to the same connected component of $G_{i}$. Otherwise, let $(u_a,\ldots,u_b)$ be a maximal subpath of $P$ that only contains vertices in $V_{\geq i+1}$. The layering $(V_0,\ldots,V_t)$ is shadow-complete and $u_{a-1},u_{b+1} \in V_i$, so $u_{a-1}u_{b+1}$ is an edge of $G$. Then, $(u_1,\ldots,u_{a-1},u_{b+1},\ldots,u_m)$ is a path from $v_1$ to $v_2$ in $G_{\geq i}$ that is shorter than $P$, which is impossible.
\end{proof}

We will construct a persistent majority coloring of $G$ recursively, starting from $G_0$ and extending it to the whole graph one layer at a time. We first define two auxiliary vertex colorings of $G$. For each $X \in {\mathcal C}$, let $\mu_X : V(X) \to [c]$ be a persistent majority coloring of $X$ with $c$ colors and $F_X$ be a persistent majority triple for $\struct{\Gamma}_{X,\mu_X}$. We call $\mu_X(v)$ the \emph{local color} of the vertex $v \in V(X)$. Define the \emph{signature} of a subset $S \subseteq V(G)$ as the set of all local colors assigned to vertices of $S$. For $i \geq 0$, the \emph{layer color} of $v \in V_i$ is $l(v) = (i\mod3)$.

Let $\gamma$ be a vertex coloring of $G$ such that two vertices have the same color if and only if they have the same local color, layer color, and signature of their respective parent cliques. Two adjacent vertices may only have the same local color if they belong to consecutive layers, and in that case they have different layer colors. Thus, $\gamma$ is a proper vertex coloring of $G$. There are at most $c$ distinct local colors, $3$ distinct layer colors, and $(c+1)^{s}$ possible signatures for parent cliques (because every shadow has size at most $s$), so $\gamma$ uses at most $3c(c+1)^{s} \leq 3(c+1)^{s+1}$ colors in total.

For $i \geq 0$, let $\gamma_{\leq i}$ denote the restriction of $\gamma$ to $V_{\leq i}$. We prove by induction on $i$ that $\struct{\Gamma}_{G_{\leq i},\gamma_{\leq i}}$ has a persistent majority triple $F_i = (f_1^i,f_2^i,f_3^i)$, starting from $i=0$. For the base case, notice that all vertices in $V_0$ have the same layer color and parent clique. Thus, $\gamma_{\leq 0}$ defines the same vertex partition as the local coloring. By Claim~\ref{claim:g0} we know that $G_0$ is connected, so $F_0 = F_{G_0}$ is a persistent majority triple for $\struct{\Gamma}_{G_{\leq 0},\gamma_{\leq 0}}$. 

Now, let $i > 0$ and suppose that $\struct{\Gamma}_{G_{\leq i-1},\gamma_{\leq i-1}}$ has a persistent majority triple $F_{i-1}$. We define $F_i(u,v,w)$ only for monochromatic triples $(u,v,w) \in V_{\leq i}^3$. We distinguish four cases.

\begin{itemize}
\item[(i)] \textbf{$u$, $v$, and $w$ all belong to $V_{\leq {i-1}}$.} In this case, we set 
\[
F_{i}(u,v,w) = F_{i-1}(u,v,w).
\]
\item[(ii)] \textbf{At least one of $u$, $v$, $w$ belongs to $V_i$.} Since $i > 0$ and $u,v,w$ are of the same color, $\kappa(u)$, $\kappa(v)$, and $\kappa(w)$ are non-empty (by Claim~\ref{claim:g0}) and have the same signature. There are three sub-cases.
\begin{itemize}
\item[(ii.a)] \textbf{$\kappa(u), \kappa(v)$, and $\kappa(w)$ are not all equal.} We set
\[
F_{i}(u,v,w) = \pi_{\kappa(u)\kappa(v)\kappa(w)}(u,v,w)
\]
where $\pi_{\kappa(u)\kappa(v)\kappa(w)}$ is as in Lemma~\ref{lem:permu-out}, with $F = F_{i-1}$.
\item[(ii.b)] \textbf{$\kappa(u), \kappa(v)$, and $\kappa(w)$ are all equal and $u,v,w \in V(X)$ for some $X \in {\mathcal C}$.} We set 
\[
F_{i}(u,v,w) = F_X(u,v,w).
\]
\item[(ii.c)] \textbf{None of the cases above applies.} Then, $\kappa(u)$, $\kappa(v)$, and $\kappa(w)$ are all equal and there exists $X \in {\mathcal C}$ such that $V(X)$ contains exactly one vertex among $u$, $v$ and $w$. We set
\[
F_{i}(u,v,w) = 
\begin{cases}
    (u,v,w) \text{ if } u \notin V(X)\\
    (w,v,u) \text{ otherwise.}
\end{cases}
\]
\end{itemize}
\end{itemize}

We prove that $F_i$ is a persistent majority triple for $\struct{\Gamma}_{G_{\leq i},\gamma_{\leq i}}$. There are three properties to verify, corresponding to (${\dagger}$), (${\dagger}{\dagger}$), and (${\dagger}{\dagger}{\dagger}$) in the alternative definition above Lemma~\ref{lem:permu-out}. We omit the proof for (${\dagger}{\dagger}$) (each tuple $F_i(u_1,u_2,u_3)$ is a permutation of $(u_1,u_2,u_3)$) because it is immediate in all cases.

\begin{claim}
    \label{claim:majinduc}
    The operation $f_1^i$ is majority.
\end{claim}

\begin{proof}
    This is equivalent to the statement that the first entry of $F_i(u,u,v)$, $F_i(u,v,u)$, and $F_i(v,u,u)$ is $u$ for any choice of $u,v \in V_{\leq i}$. This is true by the induction hypothesis for all triples that fall in case (i). Case (ii.b) is also immediate. Case (ii.c) can only happen if $u \neq v$, and then we have $v \in V(X)$, $u \notin V(X)$. Thus, the first entry must be $u$. 
    
    For the final case (ii.a), suppose that the triple is of the form $(u,u,v)$. Neither $\kappa(u)$ nor $\kappa(v)$ is empty, and $\kappa(u) \neq \kappa(v)$.  In addition, $\kappa(u)$ and $\kappa(v)$ have the same signature so $\kappa(u) \times \kappa(u) \times \kappa(v)$ contains a monochromatic triple of the form $(a,a,b)$ with $a \neq b$. We have $\pi_{\kappa(u)\kappa(u)\kappa(v)}(a,a,b) = F_{i-1}(a,a,b)$, so by the induction hypothesis the first entry of $\pi_{\kappa(u)\kappa(u)\kappa(v)}(u,u,v)$ must be $u$. The argument for triples of the form $(u,v,u)$ and $(v,u,u)$ is symmetrical. 
\end{proof}

\begin{claim}
    \label{claim:tupleinduc}
    For any relation $R_{p_1p_2} \in \struct{\Gamma}_{G_{\leq i},\gamma_{\leq i}}$ and $(u_1,v_1),(u_2,v_2),(u_3,v_3) \in R_{p_1p_2}$, there exists a coordinate permutation $\pi$ such that $F_{i}(u_1,u_2,u_3) = \pi(u_1,u_2,u_3)$ and $F_{i}(v_1,v_2,v_3) = \pi(v_1,v_2,v_3)$.
\end{claim}

\begin{proof}
    Let $l_1$ be the layer color of $u_1,u_2,u_3$ and $l_2$ be the layer color of $v_1,v_2,v_3$. 

    \begin{itemize}
        \item If $l_1 = l_2$, then each edge $u_jv_j$ belongs to some connected component $X_j \in {\mathcal C}$. Since $u_j$ and $v_j$ also share the same parent clique, $F_{i}(u_1,u_2,u_3)$ and $F_{i}(v_1,v_2,v_3)$ are determined by the same case. For cases (ii.a) and (ii.c), $F_{i}$ maps $(u_1,u_2,u_3)$ and $(v_1,v_2,v_3)$ according to the same coordinate permutation, so the claim holds. For case (i), we have $F_{i}(u_1,u_2,u_3) = F_{i-1}(u_1,u_2,u_3)$ and $F_{i}(v_1,v_2,v_3) = F_{i-1}(v_1,v_2,v_3)$. The tuples $(u_1,v_1)$, $(u_2,v_2)$, $(u_3,v_3)$ all belong to $R_{p_1p_2} \cap (V_{\leq i-1})^2 \in \struct{\Gamma}_{G_{\leq i-1},\gamma_{\leq i-1}}$, so the claim holds by the induction hypothesis. Case (ii.b) follows from a similar argument, as we have $F_{i}(u_1,u_2,u_3) = F_{X}(u_1,u_2,u_3)$, $F_{i}(v_1,v_2,v_3) = F_{X}(v_1,v_2,v_3)$, and $R_{p_1p_2} \cap (V(X))^2$ is a subset of a relation in $\struct{\Gamma}_{X,\mu_X}$.

        \item If $l_1 \neq l_2$, then each edge $u_jv_j$ has $u_j$ and $v_j$ in consecutive layers. The structure of the layer coloring ensures that either $v_j \in \kappa(u_j)$ for all $j$, or $u_j \in \kappa(v_j)$ for all $j$. Without loss of generality, we will suppose that $v_j \in \kappa(u_j)$ for all $j$. This implies in particular $v_1, v_2, v_3 \in V_{\leq i-1}$ and hence $F_{i}(v_1,v_2,v_3) = F_{i-1}(v_1,v_2,v_3)$.

        \begin{itemize}
            \item If $u_1, u_2, u_3 \in V_{\leq i-1}$, then $F_{i}((u_1,v_1),(u_2,v_2),(u_3,v_3)) = F_{i-1}((u_1,v_1),(u_2,v_2),(u_3,v_3))$ and the claim holds by the induction hypothesis.

            \item If $\kappa(u_1) = \kappa(u_2) = \kappa(u_3)$, then $v_1 = v_2 = v_3$ because a clique of $G$ contains at most one vertex of each color. In this case, $F_{i}(v_1,v_2,v_3) = \pi(v_1,v_2,v_3)$ for any coordinate permutation $\pi$. Since $F_i$ satisfies condition (${\dagger}{\dagger}$), we can select $\pi$ such that $F_{i}(u_1,u_2,u_3) = \pi(u_1,u_2,u_3)$ and the claim holds.
            
            \item Otherwise, $F_{i}(u_1,u_2,u_3)$ is determined by case (ii.a). By Lemma~\ref{lem:permu-out}, we have $F_{i}(v_1,v_2,v_3) = F_{i-1}(v_1,v_2,v_3) = \pi_{\kappa(u_1)\kappa(u_2)\kappa(u_3)}(v_1,v_2,v_3)$. Since $F_{i}(u_1,u_2,u_3) = \pi_{\kappa(u_1)\kappa(u_2)\kappa(u_3)}(u_1,u_2,u_3)$, the claim holds as well.
        \end{itemize}
    \end{itemize}
\end{proof}
\renewcommand\qedsymbol{$\square$}

By Claim~\ref{claim:majinduc} and Claim~\ref{claim:tupleinduc}, $F_i$ is a persistent majority triple for $\struct{\Gamma}_{G_{\leq i},\gamma_{\leq i}}$. By induction, this is true for all $i \geq 0$. For $i = t$, we obtain that $F_t$ is a persistent majority triple for $\struct{\Gamma}_{G,\gamma}$. The coloring $\gamma$ uses at most $3(c+1)^{s+1}$ colors and can be constructed in polynomial time.
\end{proof}

\begin{theorem}
\label{thm:topobound}
Let $H$ be a graph and ${\mathcal H}$ be a family of graphs that excludes $H$ as a topological minor. Then, there exists a constant $k = k(H)$ such that every graph in ${\mathcal H}$ has persistent majority number at most $k$. Moreover, there exists a polynomial-time algorithm that takes as input a graph $G$ and outputs either a persistent majority coloring of $G$ with at most $k$ colors or a subdivision of $H$ in $G$.
\end{theorem}

\begin{proof}
Let $G \in \mathcal H$. By Lemma~\ref{lem:rich}, $G$ is a subgraph of a graph $G'$ that has a $k_2$-rich tree decomposition $(T, \beta)$ whose bags either exclude $K_{k_2}$ or contain at most $k_2$ vertices of degree greater than $k_2$. Combining Lemma~\ref{lem:wellbehaved}, Lemma~\ref{lem:minor}, and Lemma~\ref{lem:degree}, each bag in $(T, \beta)$ induces a subgraph of $G'$ whose persistent majority number is at most $k_3 = f(k_2)$ for a certain function $f$. If $p(j)$ denotes the maximum persistent majority number of a graph that has a $j$-rich tree decomposition whose bags have persistent majority number at most $k_3$, then we have $p(0) \leq k_3$ and $p(j) \leq 3 (p(j-1)+1)^{j+1}$ for all $j>0$ by Lemma~\ref{lem:layers} and Lemma~\ref{lem:shadowcombine}. This yields the loose upper bound $\majcover{G'} \leq p(k_3) \leq (4(k_3+1))^{(k_3+1)!}$. Since $G$ is a subgraph of $G'$, we can apply Lemma~\ref{lem:wellbehaved} again and derive $\majcover{G} \leq (4(k_3+1))^{(k_3+1)!}$. A polynomial-time algorithm for computing a persistent majority coloring of $G$ that respects this bound (or a subdivision of $H$ in $G$, if we do not have the promise $G \in \mathcal H$) follows from the algorithmic part of the lemmas involved.
\end{proof}

We can now prove Theorem~\ref{thm:valhom}, which we restate below.

\thmtopo*

\begin{proof}
The theorem follows from Lemma~\ref{lem:algo} and Theorem~\ref{thm:topobound}.
\end{proof}

\section{Odd cycle homomorphism}

In this section we prove Theorem~\ref{thm:randodd}. Once again, the idea is to reduce $\ghom{C_{2k+1}}$ to an exponential number of constraint satisfaction problems over a tractable language. In order to minimize the base of the exponent, we use a relaxation of persistent majority colorings that relies on a weaker form of tractability and allow the subproblems to have overlapping solution sets. This property makes it possible to achieve runtimes of the form $O(c^n)$ with $c < 2$ via sampling.

Let $k \geq 1$. By convention, we suppose that the vertex set of $C_{2k+1}$ is $[2k+1]$ and vertices are numbered according to a fixed orientation of the cycle, i.e. all edges other than $\{1,2k+1\}$ are of the form $\{i,j\}$ with $|i-j|=1$. Define ${\mathcal S}_k = \{S_k,S_k^A,S_k^B\}$, where $S_k$ is the set of all vertices $i > 1$, $S_k^A$ is the set of all odd vertices, and $S_k^B$ is the set of all even vertices together with the vertex $1$. The family ${\mathcal S}_k$ is illustrated in Figure~\ref{fig:sets} for $k=4$. 

The crisp constraint language $\struct{\Gamma}_{C_{2k+1},{\mathcal S}_k}$ is composed of the six relations of the form $R_{U_1U_2} = \{ (u,v) \mid (u \in U_1) \land (v \in U_2) \land (uv \in E(C_{2k+1}))\}$ with $U_1,U_2 \in {\mathcal S}_k$.

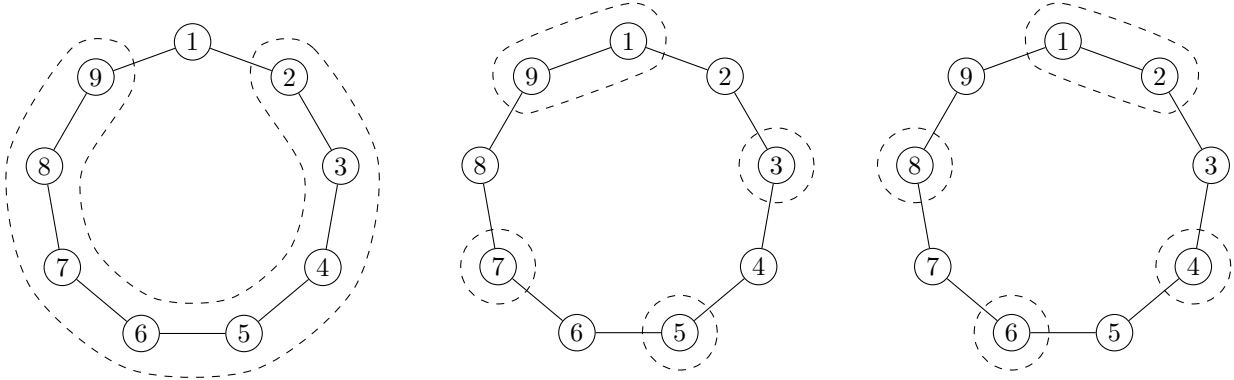
\begin{figure}
\begin{center}
\begin{minipage}[t]{0.3\textwidth}
\begin{tikzpicture}[every node/.style={circle, draw, inner sep=2pt}, scale=2]
\foreach \i in {1,...,9} {
    \node (v\i) at ({90 - 360/9*(\i-1)}:1) {$\i$};
}
\foreach \i [evaluate=\i as \j using {int(mod(\i,9)+1)}] in {1,...,9} {
    \draw (v\i) -- (v\j);
}
\draw[dashed] plot[smooth cycle] coordinates {
    (65:1.08)
    (50:1.25)
    (10:1.25)
    (-30:1.25)
    (-70:1.25)
    (-110:1.25)
    (-150:1.25)
    (-190:1.25)
    (-230:1.25)
    (-245:1.08)
    (-240:0.8)
    (-190:0.75)
    (-150:0.75)
    (-110:0.75)
    (-70:0.75)
    (-30:0.75)
    (10:0.75)
    (60:0.8)
};
\end{tikzpicture}
\end{minipage}
\hfill
\begin{minipage}[t]{0.3\textwidth}
\begin{tikzpicture}[every node/.style={circle, draw, inner sep=2pt}, scale=2]
\foreach \i in {7,5,3} {
    \node[draw, dashed, inner sep=10pt] (lab\i) at ({90 - 360/9*(\i-1)}:1) {};
}
\node[draw=none] (ghostleft) at (-190:1.25) {};
\node[draw=none] (ghostbot) at (-90:1.21) {};
\foreach \i in {1,...,9} {
    \node (v\i) at ({90 - 360/9*(\i-1)}:1) {$\i$};
}
\foreach \i [evaluate=\i as \j using {int(mod(\i,9)+1)}] in {1,...,9} {
    \draw (v\i) -- (v\j);
}
\draw[dashed] plot[smooth cycle] coordinates {
    (78:1.13)
    (90:1.25)
    (130:1.25)
    (142:1.13)
    (144:0.89)
    (130:0.75)
    (90:0.75)
    (76:0.89)
};
\end{tikzpicture}
\end{minipage}
\hfill
\begin{minipage}[t]{0.3\textwidth}
\begin{tikzpicture}[every node/.style={circle, draw, inner sep=2pt}, scale=2]
\foreach \i in {4,6,8} {
    \node[draw, dashed, inner sep=10pt] (lab\i) at ({90 - 360/9*(\i-1)}:1) {};
}
\node[draw=none] (ghostleft) at (-190:1.25) {};
\node[draw=none] (ghostbot) at (-90:1.21) {};
\foreach \i in {1,...,9} {
    \node (v\i) at ({90 - 360/9*(\i-1)}:1) {$\i$};
}
\foreach \i [evaluate=\i as \j using {int(mod(\i,9)+1)}] in {1,...,9} {
    \draw (v\i) -- (v\j);
}
\draw[dashed] plot[smooth cycle] coordinates {
    (102:1.13)
    (90:1.25)
    (50:1.25)
    (38:1.13)
    (36:0.89)
    (50:0.75)
    (90:0.75)
    (104:0.89)
};
\end{tikzpicture}
\end{minipage}
\end{center}
\caption{The sets $S_4$, $S_4^A$ and $S_4^B$ (from left to right) for $C_9$.}
\label{fig:sets}
\end{figure}

\begin{lemma}
\label{lem:majodd}
$\struct{\Gamma}_{C_{2k+1},{\mathcal S}_k}$ has a majority polymorphism.
\end{lemma}

\begin{proof}
Let $e = \{1,2k+1\}$ be the only edge in $C_{2k+1}$ with endpoints of the same parity. Call a triple $(u_1,u_2,u_3) \in [2k+1]^3$ a \emph{bad triple} if there is no set $S \in {\mathcal S}_k$ such that $\{u_1,u_2,u_3\} \subseteq S$. (Equivalently, $(u_1,u_2,u_3)$ is a bad triple if the set $\{u_1,u_2,u_3\}$ contains $1$ alongside two vertices strictly greater than $1$, exactly one of which is even.) For $a,b,c \in \mathbb{N}$ with $a \geq b$, let 
\[
\prescript{a}{}{\lceil} c \rfloor_b =
\begin{cases}
    a \text{ if } c > a\\
    b \text{ if } c < b\\
    c \text{ otherwise }
\end{cases}
\]
and define $f : [2k+1]^3 \to [2k+1]$ as the unique operation such that
\begin{itemize}
    \item $f$ is invariant under any permutation of its arguments;
     \item $f(a,b,c) = \text{median}(a,b,c)$ if $a$, $b$, $c$ all have the same parity;
    \item $f(a,b,c) = \prescript{a}{}{\lceil} 2k + 3 - c \rfloor_b$ if $a \geq b$ and the parity of $c$ differs from that of both $a$ and $b$.
\end{itemize}
The operation $f$ is majority. We will prove that $f$ is a polymorphism of $\struct{\Gamma}_{C_{2k+1},{\mathcal S}_k}$ by combining the following claims.

\begin{claim}
\label{claim:fpreserves}
If $S \in {\mathcal S}_k$ and $u_1,u_2,u_3 \in S$, then $f(u_1,u_2,u_3) \in S$.
\end{claim}

\renewcommand\qedsymbol{$\blacksquare$}
\begin{proof}
Observe that $f$ can only map a triple $(u_1,u_2,u_3)$ to $1$ if $1 \in \{u_1,u_2,u_3\}$. Therefore, the claim holds for $S = S_k$. If $S = S_k^A$ instead, then $f(u_1,u_2,u_3) = \text{median}(u_1,u_2,u_3) \in \{u_1,u_2,u_3\}$ and the claim holds as well. Finally, if $S = S_k^B$ then either:
\begin{itemize}
    \item at least two vertices among $u_1,u_2,u_3$ are equal and $f(u_1,u_2,u_3) \in \{u_1,u_2,u_3\}$ by majority, or
    \item all three vertices are even and $f(u_1,u_2,u_3) = \text{median}(u_1,u_2,u_3) \in \{u_1,u_2,u_3\}$, or
    \item one vertex is $1$ and the other two are even, in which case $f(u_1,u_2,u_3)$ is the largest even vertex among $u_1,u_2,u_3$.
\end{itemize}
In all cases, $f(u_1,u_2,u_3) \in S$.
\end{proof}

\begin{claim}
\label{claim:e}
    If $u_1v_1$, $u_2v_2$, $u_3v_3$ are edges of $C_{2k+1}$ such that $e \notin \{u_1v_1,u_2v_2,u_3v_3\}$, then $f(u_1,u_2,u_3)f(v_1,v_2,v_3)$ is an edge of $C_{2k+1}$.
\end{claim}

\begin{proof}
The only edge of $C_{2k+1}$ whose endpoints are not consecutive numbers in $\mathbb{N}$ is $e$. Therefore, if $e$ does not appear in $u_1v_1, u_2v_2, u_3v_3$ then for all $i \leq 3$ we have either $v_i = u_i - 1$ or $v_i = u_i + 1$. In particular, $u_i$ and $v_i$ always have opposite parity. We distinguish two cases.

\begin{itemize}
\item $u_1,u_2,u_3$ all have the same parity. Then, the same is true for $v_1,v_2,v_3$ and we have $f(u_1,u_2,u_3) = \text{median}(u_1,u_2,u_3)$, $f(v_1,v_2,v_3) = \text{median}(v_1,v_2,v_3)$. If there exist two indices $i,j$ such that $u_i = u_j$, then $|v_i - v_j| \in \{0,2\}$ and hence one of $v_i,v_j$ is the median of $(v_1,v_2,v_3)$. It follows that $f(u_1,u_2,u_3)f(v_1,v_2,v_3)$ is either $u_iv_i$ or $u_jv_j$. In both cases, $f(u_1,u_2,u_3)f(v_1,v_2,v_3)$ is an edge of $C_{2k+1}$. On the other hand, if $u_1,u_2,u_3$ are all distinct then $u_i < u_j$ implies $v_i \leq v_j$ (for any choice of $i,j$). Consequently, $\text{median}(u_1,u_2,u_3) = u_i$ implies $\text{median}(v_1,v_2,v_3) = v_i$ and hence $f(u_1,u_2,u_3)f(v_1,v_2,v_3)$ is an edge of $C_{2k+1}$ as well.

\item $u_1,u_2,u_3$ do not have the same parity. Again, the same is true for $v_1,v_2,v_3$. Since $f$ is invariant under any permutation of its arguments, we can assume without loss of generality that $u_1,u_2$ have the same parity, $u_1 \geq u_2$, and $v_1 \geq v_2$. Let $w_u = 2k + 3 - u_3$ and $w_v = 2k + 3 - v_3$. The vertices $u_3$ and $v_3$ are consecutive numbers in $\mathbb{N}$, so $w_u$ and $w_v$ are consecutive as well and $w_uw_v$ is an edge of $C_{2k+1}$. We have $f(u_1,u_2,u_3) = \prescript{u_1}{}{\lceil} w_u \rfloor_{u_2}$ and $f(v_1,v_2,v_3) = \prescript{v_1}{}{\lceil} w_v \rfloor_{v_2}$, so if $f(u_1,u_2,u_3)f(v_1,v_2,v_3)$ is not an edge of $C_{2k+1}$ then either $w_u \notin [u_2,u_1]$ or $w_v \notin [v_2,v_1]$. If $w_u > u_1$, then $w_v \geq w_u - 1 > u_1 - 1 \geq v_1 - 2$. Since $w_v$ and $v_1$ have the same parity, we deduce $w_v \geq v_1$ and hence $f(u_1,u_2,u_3)f(v_1,v_2,v_3) = u_1v_1$ is an edge of $C_{2k+1}$. The cases $w_u < u_2$, $w_v < v_2$ and $w_v > v_1$ are symmetrical.
\end{itemize}
\end{proof}

\begin{claim}
\label{claim:bad}
    If $u_1v_1$, $u_2v_2$, $u_3v_3$ are three edges of $C_{2k+1}$ such that neither $(u_1,u_2,u_3)$ nor $(v_1,v_2,v_3)$ is a bad triple, then $f(u_1,u_2,u_3)f(v_1,v_2,v_3)$ is an edge of $C_{2k+1}$.
\end{claim}

\begin{proof} By Claim~\ref{claim:e}, it is sufficient to prove Claim~\ref{claim:bad} in the case where $\{u_1v_1,u_2v_2,u_3v_3\}$ contains $e$. Since $f$ is invariant under permutation of its arguments, we can further assume that $u_1 = 1$ and $v_1 = 2k+1$. We proceed again by case analysis.

\begin{itemize}
    \item There exists $i \in \{2,3\}$ such that $u_i = 1$. Then, we have either $v_i = 2$ or $v_i = 2k+1$. If $v_i = 2k+1$ then $f(u_1,u_2,u_3)f(v_1,v_2,v_3) = e$ (because $f$ is majority) and the claim follows. If $v_i = 2$ instead, then $v_1$, $v_2$, $v_3$ do not all have the same parity and either $f(v_1,v_2,v_3) = 2$ (if two vertices among $v_1$, $v_2$, $v_3$ are even) or $f(v_1,v_2,v_3) = 2k+1$ (if two vertices among $v_1$, $v_2$, $v_3$ are odd). In both cases, $f(u_1,u_2,u_3)f(v_1,v_2,v_3)$ is an edge of $C_{2k+1}$.
    \item Both $u_2$ and $u_3$ are odd and strictly greater than $1$. We have three sub-cases.
        \begin{itemize}
        \item Both $v_2$ and $v_3$ are even. We can assume without loss of generality that $u_2 \geq u_3$ and $v_2 \geq v_3$. We have $f(u_1,u_2,u_3) = \text{median}(1,u_2,u_3) = u_3$ and $f(v_1,v_2,v_3) = \prescript{v_2}{}{\lceil} 2k + 3 - (2k+1) \rfloor_{v_3} = v_3$, hence $f(u_1,u_2,u_3)f(v_1,v_2,v_3) = u_3v_3$ is an edge of $C_{2k+1}$.
        \item Exactly one vertex among $v_2,v_3$ is even. We can assume that $v_2$ is even and $v_3$ is odd. Both $u_3$ and $v_3$ are odd and $u_3 \neq 1$, so $u_3 = 2k+1$ and $v_3 = 1$. This implies that $(v_1,v_2,v_3)$ is a bad triple, so this case is impossible.
        \item Both $v_2$ and $v_3$ are odd. Then, $v_2 = v_3 = 1$ and $u_2 = u_3 = 2k+1$. By majority we deduce $f(u_1,u_2,u_3)f(v_1,v_2,v_3) = \{2k+1,1\}$, which is an edge of $C_{2k+1}$.
        \end{itemize}
    \item Both $u_2$ and $u_3$ are even. Then, $v_1$ and $v_2$ are odd and we can assume $u_2 \geq u_3$, $v_2 \geq v_3$. In this case we have $f(u_1,u_2,u_3) = \prescript{u_2}{}{\lceil} 2k + 3 - 1 \rfloor_{u_3} = u_2$ and $f(v_1,v_2,v_3) = \text{median}(2k+1,v_2,v_3) = v_2$, so $f(u_1,u_2,u_3)f(v_1,v_2,v_3) = u_2v_2$ is an edge of $C_{2k+1}$.
\end{itemize}

If none of the cases above apply, then both $u_2$ and $u_3$ are strictly greater than $1$ and exactly one of them is even. This is impossible because $(u_1,u_2,u_3)$ is not a bad triple.

\end{proof}

\renewcommand\qedsymbol{$\square$}

We can now finish the proof of the lemma. Let $R \in \struct{\Gamma}_{C_{2k+1},{\mathcal S}_k}$ and $(u_1,v_1), (u_2,v_2), (u_3,v_3)$ be three tuples in $R$. These tuples correspond to edges in $C_{2k+1}$ and no set in ${\mathcal S}_k$ contains a bad triple, so we deduce from Claim~\ref{claim:bad} that $f(u_1,u_2,u_3)f(v_1,v_2,v_3)$ is an edge of $C_{2k+1}$. Moreover, by Claim~\ref{claim:fpreserves} we have that $f(u_1,u_2,u_3)$ (resp. $f(v_1,v_2,v_3)$) belongs to the same set as $u_1,u_2,u_3$ (resp. $v_1,v_2,v_3$). This implies $(f(u_1,u_2,u_3),f(v_1,v_2,v_3)) \in R$ and hence $f$ is a majority polymorphism of $\struct{\Gamma}_{C_{2k+1},{\mathcal S}_k}$.
\end{proof}

\thmodd*

\begin{proof}
Let $G$ be a graph. Consider the following algorithm, which we call Algorithm B. Let $p : {\mathcal S}_k \to [0,1]$ be the probability distribution over ${\mathcal S}_k$ such that $p(S_k) = (2k-1)/(2k+1)$ and $p(S_k^A) = p(S_k^B) = 1/(2k+1)$.
For each vertex $v \in V(G)$, draw $S \sim p$ and set $L(v) = S$. Cast the problem of finding a homomorphism $g$ from $G$ to $C_{2k+1}$ such that $g(v) \in L(v)$ for all $v \in V(G)$ as a \plainVCSP\ instance $I$, with one variable $x_v$ with domain $[2k+1]$ for each $v \in V(G)$ and a term $\phi_{uv}(x_u,x_v)$ for each edge $uv \in E(G)$. The cost function $\phi_{uv} : [2k+1]^2 \to \{0,\infty\}$ is such that $\phi_{uv}(a,b) = 0 \iff (a \in L(u)) \land (b \in L(v)) \land (ab \in E(C_{2k+1}))$. Construct the linear program $\sa{2}{3}{I}$ and compute its optimum $s^*$. If $s^* = 0$ then return YES. Otherwise, return NO.

\renewcommand\qedsymbol{$\blacksquare$}
\begin{claim}
\label{claim:corr}
If Algorithm B returns YES, then there exists a homomorphism from $G$ to $C_{2k+1}$.
\end{claim}

\begin{proof}
By Lemma~\ref{lem:majodd}, the crisp constraint language $\struct{\Gamma}_{C_{2k+1},{\mathcal S}_k}$ has a majority polymorphism $f$. The valued constraint language $\Gamma$ of $I$ is the all-zero valuation of $\struct{\Gamma}_{C_{2k+1},{\mathcal S}_k}$, and it is easy to verify that $\Gamma$ admits as fractional polymorphism the mapping $\omega$ that gives weight $1$ to $f$ and weight $0$ to all other ternary polymorphisms of $\struct{\Gamma}_{C_{2k+1},{\mathcal S}_k}$. By Theorem~\ref{thm:vcspwidth}, $\opt{I}$ is exactly the optimum of $\sa{2}{3}{I}$. Therefore, if the algorithm above returns YES then $\opt{I} = 0$ and hence there exists a homomorphism from $G$ to $C_{2k+1}$.
\end{proof}

\begin{claim}
\label{claim:prob}
If there exists a homomorphism from $G$ to $C_{2k+1}$, then Algorithm B returns YES with probability at least $(\alpha_k)^{-n}$.
\end{claim}

\begin{proof}
Let $g$ be a homomorphism from $G$ to $C_{2k+1}$. For each $u \in [2k+1]$, let $\mu(u) = |\{ v \in V(G) \mid g(v) = u \}|$. Clearly, there exists $u^* \in [2k+1]$ such that $\mu(u^*) \leq n/(2k+1)$. Due to the automorphisms of $C_{2k+1}$, we can assume without loss of generality that $u^* = 1$. The probability that $g(v) \in L(v)$ for all $v \in V(G)$ is
\begin{align*}
\prod_{v \in V(G)} \text{Pr}(g(v) \in L(v)) &= \prod_{u \in V(H)} \left( \sum_{S \in {\mathcal S}_{k} \, : \, u \in S} p(S) \right)^{\mu(u)}\\
&= \left( \frac{2}{2k+1} \right)^{\mu(1)} \left( \frac{2k}{2k+1} \right)^{n-\mu(1)}\\
&\geq \left( \left( \frac{2}{2k+1} \right)^{\frac{1}{2k+1}} \left( \frac{2k}{2k+1} \right)^{\frac{2k}{2k+1}} \right)^n\\
&= (\alpha_k)^{-n}
\end{align*}
where the inequality comes from the fact that $(2/(2k+1))^x(2k/(2k+1))^{1-x}$ is a nonincreasing function of $x$ for $0 \leq x \leq 1/(2k+1)$.
\end{proof}
\renewcommand\qedsymbol{$\square$}

By Claim~\ref{claim:corr} and Claim~\ref{claim:prob}, invoking $\lceil \ln(1/2) /\ln(1-(\alpha_k)^{-n})\rceil  + 1 = O((\alpha_k)^n)$ times Algorithm B and returning the disjunction of the outputs (seen as Boolean values) yields a Monte Carlo algorithm that returns the correct answer with probability $> 1/2$. The runtime of Algorithm B is polynomial in $n$ (and $k$), so the theorem follows.
\end{proof}

We could verify through computer search that $\struct{\Gamma}_{C_{9},{\mathcal S}_4}$ does not have a persistent majority support, despite having a majority polymorphism. This explains why Theorem~\ref{thm:randodd} only applies to $\ghom{C_{2k+1}}$ and not its optimization counterpart $\valhom{C_{2k+1}}$. We also observed that the family ${\mathcal S}_4$ is maximal, in the sense that adding any vertex to a set (or including any additional set that is not contained in either $S_4$, $S_4^A$, or $S_4^B$) always results in an NP-hard constraint language. This suggests that improving upon Theorem~\ref{thm:randodd} will require a more sophisticated approach than pure sampling.

\section{Concluding remarks}

An obvious open question would be to find a graph property that smoothly generalizes both bounded persistent majority number and bounded (extended) cliquewidth, and implies the existence of an algorithm for $\ghom{\mathcal H}$ with single-exponential dependency on $n+h$. We believe that settling the complexity of $\ghom{\mathcal H}$ when $\mathcal H$ has bounded twin-width could shed some light on this question, as a single-exponential algorithm in this case would cover both bounded track number (see Bonnet et al.~\cite{DBLP:journals/combtheory/BonnetGKTW22}) and bounded cliquewidth.

Another way to pursue this line of research would be to shift the focus from $\ghom{\mathcal H}$ to $\valhom{\mathcal H}$. The increased expressivity of $\valhom{\mathcal H}$ makes complexity lower bounds much easier to derive. For example, any algorithm for $\valhom{\mathcal H}$ also works on the subgraph closure of $\mathcal H$; under the ETH, this observation rules out many properties that allow dense graphs (such as bounded cliquewidth) from admitting single-exponential algorithms. Technicalities related to homomorphic equivalence are also largely removed from consideration, which may result in a cleaner complexity landscape.

Our last question concerns the $\subiso{\mathcal H}$ problem, which asks for the existence of an \emph{injective} homomorphism from $G$ to $H \in {\mathcal H}$. If $\mathcal H$ excludes a topological minor, then coupling Theorem~\ref{thm:valhom} with a standard color-coding argument yields an algorithm for $\subiso{\mathcal H}$ that runs in time $2^{O(n \log n)}h^{O(1)}$ and polynomial space. This matches the complexity of the best known algorithm for this problem, due to combined results of Pilipczuk and Siebertz~\cite{DBLP:journals/jctb/PilipczukS21} and D\c{e}bski et al.~\cite{Debski2021Improved}. These two algorithms work in a somewhat different fashion, although we note that both rely on branching schemes that share several critical ingredients: a variant of the planar product structure theorem, the Grohe-Marx structure theorem for excluded topological minors, and color-coding. Can either of these approaches be refined to yield a single-exponential time, polynomial-space algorithm for $\subiso{\mathcal H}$? 

\bibliographystyle{plain}
\bibliography{sample}

\end{document}